\documentclass[a4paper,USenglish,pdfa,cleveref, thm-restate, autoref, numberwithinsect]{lipics-v2021}
\usepackage[dvipsnames]{xcolor}
\usepackage{xspace}
\usepackage{booktabs}

\usepackage{tikz}
\usetikzlibrary{
    arrows,
    arrows.meta,
    calc,
    decorations.pathreplacing,
    intersections,
    matrix,
    positioning,
}

\graphicspath{{./figures/}}

\newcommand{\R}{\ensuremath{\mathbb{R}}\xspace}

\renewcommand{\epsilon}{\varepsilon}
\newcommand{\eps}{\ensuremath{\varepsilon}\xspace}

\newcommand{\etal}{et al.\xspace}

\DeclareMathOperator{\polylog}{polylog}

\newcommand{\A}{\ensuremath{\mathcal{A}}\xspace}

\newcommand{\T}{\ensuremath{\mathcal{T}}\xspace}

\newcommand{\kNN}{\ensuremath{k\text{-NN}}\xspace}
\newcommand{\Vor}{\ensuremath{\mathit{Vor}}}
\newcommand{\Dk}{\ensuremath{\mathcal{D}^{k}}}

\newcommand{\LL}{\ensuremath{\mathcal{L}}\xspace}

\newcommand{\Matousek}{Matou{\v s}ek\xspace}

\newcommand{\Otilde}{\tilde{O}}

\newcommand{\thmheadfont}{\textcolor{darkgray}{$\blacktriangleright$}\nobreakspace\sffamily\bfseries}
\newenvironment{repeatenv}[2]%
  {\noindent {\thmheadfont #1~\ref{#2}.}\ \slshape}
  {\normalfont}

\title{Chromatic $k$-Nearest Neighbor Queries}

\newcommand {\grantsponsor} [3] {\href{#3}{#2}}
\newcommand {\grantnum} [2] {#2}

\author{Thijs van der Horst}{Department of Information and Computing Sciences, Utrecht University, The Netherlands}{t.w.j.vanderhorst@uu.nl}{}{}
\author
{Maarten L\"{o}ffler}
{Department of Information and Computing Sciences, Utrecht University, The Netherlands
\and \url{https://www.uu.nl/staff/MLoffler/Profile}}
{m.loffler@uu.nl}
{}
{Partially supported by the \grantsponsor{NWO}{Dutch Research Council (NWO)}{https://www.nwo.nl/} under the project numbers \grantnum{NWOb}{614.001.504} and \grantnum{NWOc}{628.011.005}.}

\author{Frank Staals}{Department of Information and Computing Sciences, Utrecht University, The Netherlands}{f.staals@uu.nl}{}{}

\authorrunning{T. van der Horst, M. L\"offler and F. Staals}

\Copyright{Thijs van der Horst, Maarten L\"offler and Frank Staals}



\acknowledgements{We would like to thank an anonymous reviewer for the
  randomized solution presented in
  Section~\ref{subsub:rand_finding_disk}, which led to our current
  solution for finding $\Dk_2(q)$ in
  Section~\ref{subsub:deter_finding_disk}.}

\nolinenumbers

\hideLIPIcs

\begin{document}

\maketitle

\begin{abstract}
  Let $P$ be a set of $n$ colored points. We develop efficient data
  structures that store $P$ and can answer chromatic $k$-nearest
  neighbor ($k$-NN) queries. Such a query consists of a query point
  $q$ and a number $k$, and asks for the color that appears most
  frequently among the $k$ points in $P$ closest to $q$. Answering such
  queries efficiently is the key to obtain fast $k$-NN
  classifiers. Our main aim is to obtain query times that are
  independent of $k$ while using near-linear space. 
  
  We show that this is possible using a combination of two data
  structures. The first data structure allow us to compute a region
  containing exactly the $k$-nearest neighbors of a query point $q$,
  and the second data structure can then report the most frequent
  color in such a region. This leads to linear space data structures
  with query times of $O(n^{1 / 2} \log n)$ for points in
  $\mathbb{R}^1$, and with query times varying between
  $O(n^{2/3}\log^{2/3} n)$ and $O(n^{5/6} \polylog n)$, depending on the distance
  measure used, for points in $\mathbb{R}^2$. Since these query times
  are still fairly large we also consider approximations. If we are
  allowed to report a color that appears at least $(1-\varepsilon)f^*$
  times, where $f^*$ is the frequency of the most frequent color, we
  obtain a query time of
  $O(\log n + \log\log_{\frac{1}{1-\varepsilon}} n)$ in
  $\mathbb{R}^1$ and expected query times ranging between
  $\tilde{O}(n^{1/2}\varepsilon^{-3/2})$ and
  $\tilde{O}(n^{1/2}\varepsilon^{-5/2})$ in $\mathbb{R}^2$
  using near-linear space (ignoring polylogarithmic factors).
\end{abstract}

\section{Introduction}
\label{sec:Introduction}

\begin{figure}[tb]
  \centering
  \includegraphics[height=3.9cm,page=1]{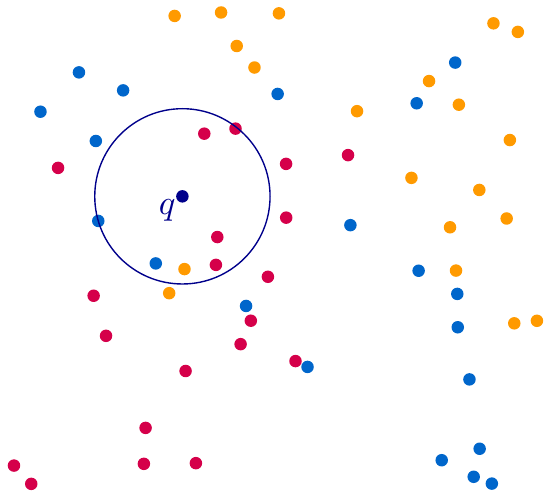}
  \quad
  \includegraphics[height=3.9cm,page=3]{knn}
  \quad
  \includegraphics[height=3.9cm]{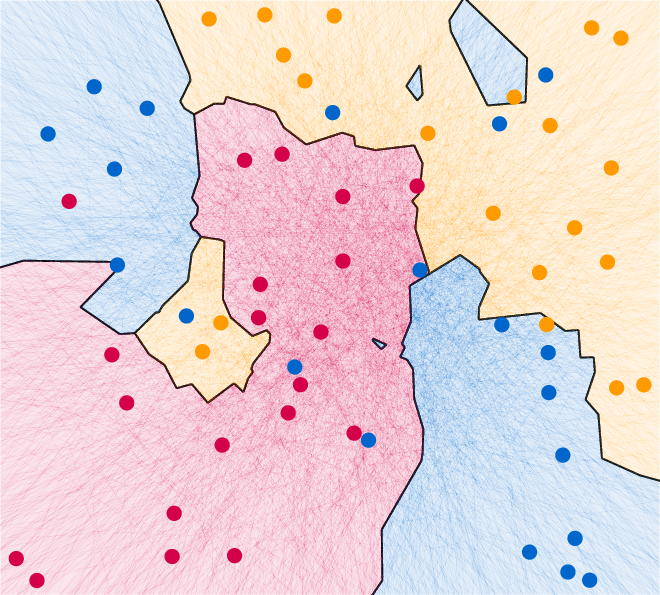}
  \caption{(left) A set of input points from three different classes
    (colors). The class of a query point $q$ is determined by the
    labels of its $k$ nearest neighbors (with $k=7$ as shown here $q$
    is classified as red). (middle) The color partition for $k=1$.
    (right) The color partition for $k=3$.
    }
  \label{fig:knn}
\end{figure}

One of the most popular approaches for classification problems is to
use a $k$-Nearest Neighbor (\kNN)
classifier~\cite{aggarwal2014data,cover1967nearest,henley96neares_neigh_class_asses_consum_credit_risk,law05adapt_neares_neigh_class_algor_data_stream}. In
a \kNN classifier the predicted class of a query item $q$ is taken to
be the most frequently appearing class among the $k$ items most
similar to $q$. One can model this as a geometric problem in which the
input items are represented by a set $P$ of $n$ colored points in
$\R^d$: the color of the points represents their class, and the
distance between points measures their similarity. The goal is then to
store $P$ so that one can efficiently find the color (class) $c^*$
most frequently occurring among the $k$ points in $P$ closest to a
query point $q$. See Figure~\ref{fig:knn}(left). We refer to such
queries as \emph{chromatic \kNN queries}. To answer such queries, \kNN
classifiers often store $P$ in, e.g., a kd-tree and answer queries by
explicitly reporting the $k$ points closest to $q$, scanning through
this set to compute the most frequently occurring
color~\cite{aggarwal2014data}. Unfortunately, for many distance
measures (including the Euclidean distance) such an approach has no
guarantees on the query time other than the trivial $O(n)$ time
bound. See Figure~\ref{fig:kd_tree_worst_case} for an
illustration. Even assuming that the dependency on $n$ during the
query time is small (e.g. when the points are nicely
distributed~\cite{friedman77logarithmic_classification}), the approach
requires $\Theta(k)$ time to explicitly process all $k$ points closest
to $q$, whereas the desired output is only a single value: the most
frequently appearing color. Hence, our main goal is to design a data
structure to store $P$ that has sublinear query time in terms of both
$n$ and $k$, while still using only small space. We focus our
attention on the $L_2$ (Euclidean) distance, and the $L_\infty$ distance 
metrics. Most of our ideas extend to more general distance measures and
higher dimensions as well. However, already in the restricted settings
presented here, designing data structures that provide guarantees on
the space and query time turns out to be a challenging task.

\begin{figure}
  \centering
  \includegraphics[width=0.5\textwidth
                  ,clip,trim=0cm 1cm 0cm 1cm
                  ]{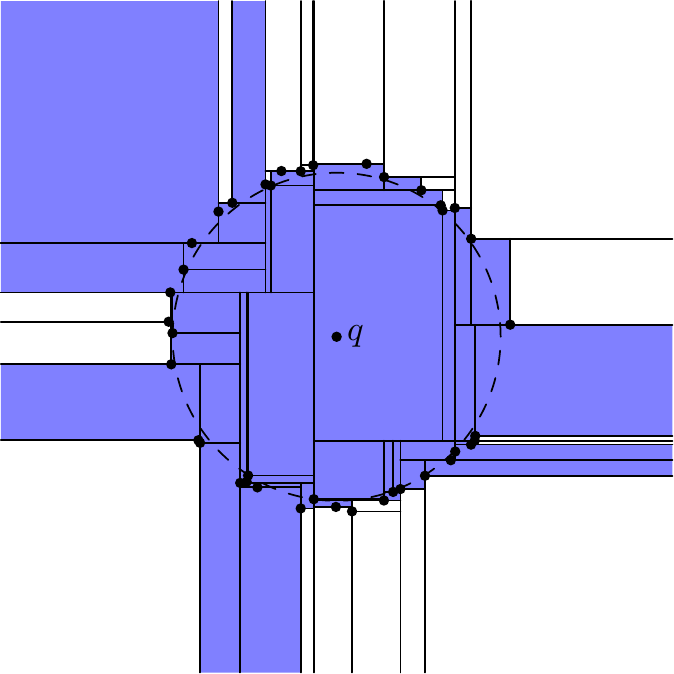}
  \caption{An example configuration in which querying a kd-tree for
    the $k$ nearest neighbors (in terms of the Euclidean distance) may
    visit a linear number of cells (shown in blue).}
  \label{fig:kd_tree_worst_case}
\end{figure}

If the value $k$ is fixed in advance, one possible solution is to build the
$k^\mathrm{th}$-order Voronoi diagram $\Vor_k(P)$ of $P$, and
preprocess it for point location queries. The
\emph{$k^\mathrm{th}$-order Voronoi diagram} is a partition of $\R^d$
into maximal cells for which all points in a cell (a \emph{Voronoi
  region}) $V_{k,P}(S)$ have the same set $S$ of $k$ closest points
from $P$. Hence, in each Voronoi region there is a fixed color that
occurs most frequently among $S$. See Figure~\ref{fig:knn}(right) for an
illustration. By storing $\Vor_k(P)$ in a data structure for efficient
(i.e. $O(\log n)$ time) point location queries we can also answer
chromatic \kNN queries efficiently. However, unfortunately $\Vor_k(P)$
may have size $\Theta(k(n-k))$~\cite{lee82voronoi} (for points in
$\R^2$ and the $L_2$
distance). For other $L_m$ distances the diagram
is similarly
large~\cite{liu15neares_neigh_voron_diagr_revis,bohler15voron}. Hence,
we are interested in solutions that use less, preferably near-linear, space.

The only result on the theory of chromatic \kNN queries that we are
aware of is that of Mount
\etal~\cite{DBLP:journals/comgeo/MountNSW00}. They study the problem
in the case that we measure distance using the Euclidean metric and
that the number of colors $c$, as well as the parameter $k$, are small
constants. Mount~\etal state that it is unclear how to obtain a query
time independent of $k$, and instead analyze the query times in terms
of the \emph{chromatic density} $\rho$ of a query $q$. The chromatic
density is a term depending on the distance from the query point $q$
to the $k^\mathrm{th}$ nearest neighbor of $q$, and the distance from
$q$ to the first point at which the answer of the query would change
(e.g. the $(k+t)^\mathrm{th}$ nearest neighbor of $q$ for some
$t > 0$). The intuition is that if many points near $q$ have the same
color, queries should be easier to answer than when there are multiple
colors with roughly the same number of points. The chromatic
density term models this. Their main result is a linear space data
structure for points in $\R^d$ that supports
$O(\log^2 n + (1/\rho)^d \log (1/\rho))$ time queries. We aim for
bounds only in terms of combinatorial properties (i.e. $n$, $c$, and
$k$) and allow the number of colors, as well as the parameter $k$, to
depend on $n$. Our results are particularly relevant when $k$ and $c$
are large compared to $n$.

\subparagraph{Our approach.} Our main idea is to answer a query in two
steps. (1) We identify a region $\Dk_m(q)$ that contains exactly the
set $\kNN_m(q)$ of the $k$ sites closest to $q$ according to distance
metric $L_m$. (2) We then find the \emph{mode color} $c^*$; that is,
the most frequently occurring color among the points in the region
$\Dk_m(q)$. This way, we never have to explicitly enumerate the set
$\kNN_m(q)$. We will design separate data structures for these two
steps. Our data structure for step (1) will find the smallest metric
disk $\Dk_m(q)$ containing $\kNN_m(q)$ centered at $q$. We refer to
such a query as an \emph{range finding} query. If the distance used is
clear from the context we may write $\kNN(q)$ and $\Dk(q)$ instead.

\subparagraph{Range mode queries.} The data structure in step (2) answers
so-called \emph{range mode} queries. For these data structures we exploit and extend the result of Chan
\etal~\cite{chan14linear_space_data_struc_range}. They show an array
$A$ with $n$ entries can be stored in a linear space data structure
that allows reporting the mode of a query range (interval) $A[i..j]$
in $O(n^{1/2})$ time. Furthermore, points in $\R^d$ can be stored in
$O(n\polylog n + r^{2d})$ space so that range mode queries with axis-aligned 
orthogonal ranges can be answered in $O((n/r)\polylog n)$
time. Here, $r \in [1,n]$ is a user-choosable parameter. In
particular, setting $r=\lceil n^{1/2d} \rceil$ yields an
$O(n\polylog n)$ space solution with $O(n^{1-1/2d}\polylog n)$ query
time. Range mode queries with halfspaces can be answered in
$O((n/r)^{1-1/d^2} + \polylog n)$ time using $O(nr^{d-1})$
space~\cite{chan14linear_space_data_struc_range}. Range mode queries
in arrays have also been considered in an approximate
setting~\cite{DBLP:conf/stacs/BoseKMT05}. The goal is then to report
an element that appears sufficiently often in the range.

\begin{table}
\centering
\begin{tabular}{ll lll}
    \toprule
    Dimension & Metric & Preprocessing time & Space & Query time \\
    \midrule
    $d = 1$ & $L_m$ & $O(n^{3/2}\log n)$ & $O(n)$ & $O(n^{1/2} \log n)$ \\
    \midrule
    $d = 2$ & $L_m$ & $\Otilde(n^{2/(4+\delta)})$ & $O(n)$ & $\Otilde(n^{1-1/(12 + 3\delta)})$ \\
    \midrule
    $d \geq 3$ & $L_m$ & $\Otilde(n^{2-1/d} + n^{1+d/(2d-2+\delta)})$ & $O(n)$ & $\Otilde(n^{1-1/((d+1)(2d-2+\delta))})$ \\
    \midrule
    \multirow{3}{*}{$d \geq 2$} & \multirow{2}{*}{$L_\infty$} &
                                                $O(n^{1+d/(d+1)})$ & $O(n)$ & $O(n^{1-1/(d+1) + \delta})$ \\
    & &                                         $\Otilde(n^{1+d/(d+1)})$ & $O(n\log^{d-1} n)$ & $O((n\log^{d-1} n)^{1-1/(d+1)})$ \\
    \cmidrule{2-5}
    & $L_2$ & $\Otilde(n^{2-1/d} + n^{1+d/(d+1)})^*$ & $O(n)$ &
                                                                    $\Otilde(n^{1-(d-1)/d(d+1)})$ \\
    \bottomrule
\end{tabular}
\caption{Our results for exact chromatic $k$-nearest neighbors
  problems. Bounds marked with $^*$ are expected bounds.
  The general $L_m$ metric bounds hold for $m = O(1)$.
}
\label{tbl:exact_solutions}
\end{table}

\subparagraph{Results and organization.} 
Refer to Table~\ref{tbl:exact_solutions} for an overview of our exact solutions.
We first consider the problem
for $n$ points in $\R^1$. In this setting, we develop an optimal linear space
data structure that can find $\Dk_m(q)$ in $O(\log n)$ time, for any
$m \geq 1$ (Section~\ref{sub:1d_find_region}). We then use Chan
\etal~\cite{chan14linear_space_data_struc_range}'s data structure to
report the mode color in $\Dk_m(q)$. Since we present all of our data
structures in the pointer machine model augmented with real-valued
arithmetic, this yields an $O(n^{1/2}\log n)$ time query
algorithm.  
There is a conditional $\Omega(n^{1/2-\delta})$ lowerbound for
chromatic \kNN queries using linear space and $O(n^{3 / 2})$
preprocessing time (refer to Section~\ref{sec:Lower_bounds}) so this
result is likely near-optimal. In
Section~\ref{sec:finding_the_range_2d} we present our data structure
for finding $\Dk_m(q)$ in $\R^2$. For the $L_\infty$ metric
we show that we can essentially find $\Dk_m(q)$ using a binary search
on the radius of the disk, and thus there is a simple $O(n\log n)$
size data structure that allows us to find the range $\Dk(q)$ in
$O(\log^2 n)$ time (or $O(n^\delta)$ time, for an arbitrarily small
$\delta > 0$, in case of a linear space structure). For the $L_2$
metric, we can no longer easily access a discrete set of candidate
radii. It is tempting to therefore replace the binary search by a
parametric
search~\cite{megiddo1983parametric,megiddo1979parametric}. However,
the basic such approach squares the $\Otilde(n^{1/2})$ time required
to solve the decision problem and is thus not applicable. The full
strategy requires a way to generate independent comparisons; typically
by designing a parallel decision
algorithm~\cite{megiddo1983parametric}. Neither task is
straightforward to achieve. Instead, we show that we can directly
adapt the query algorithm for answering semi-algebraic range
queries~\cite{agarwal13semialgebraic} (essentially the decision
algorithm in the approach sketched above) to find $\Dk_2(q)$ in
$O(n^{1/2} \polylog n)$ time. Unfortunately, the final query time for
chromatic \kNN queries is dominated by the $O(n^{5/6}\polylog n)$ time
range mode queries, for which we use a slight variation of the data
structure of Chan~\etal~\cite{chan14linear_space_data_struc_range}. We
briefly discuss these results in Section~\ref{sec:range_mode_2d}, We
do show that the data structure can be constructed in $O(n^{5/3})$
expected time, rather than the straightforward $O(n^2)$ time
implementation that directly follows from the description of Chan
\etal. For the $L_\infty$ metric we can answer range mode
queries in $O(n^{3/4}\polylog n)$ time using Chan~\etal's data
structure for orthogonal ranges. However, we show that we can exploit
that the ranges are squares and use a cutting-based approach (similar
to the one used for the $L_2$ distance) to answer queries in
$O(n^{2/3}\log^{2/3}n)$ time instead. If we wish to reduce the space
from $O(n\log n)$ to linear the query time becomes
$O(n^{2/3+\delta})$.

Since the query times are still rather large, we then turn our
attention to approximations; refer to
Table~\ref{tbl:approximate_solutions} for an overview. If $f^*$ is the
frequency of the mode color among $\kNN(q)$ then our data structures
may return a color that appears at least $(1-\eps)f^*$ times. In case
of the $L_2$ distance our data structure presented in
Section~\ref{sec:approximation} now achieves roughly
$\Otilde(n^{1/2}\eps^{-3/2})$ query time, where the $\Otilde$ notation
hides polylogarithmic factors of $n$ and $\eps$. The main idea is that
approximate levels in the arrangement of distance functions have
relatively low complexity, and thus we can store them to efficiently
answer approximate range mode queries.
We show how to generalize our
results to higher dimensions and other distance metrics in
Section~\ref{sec:Extensions}.

\begin{table}
\centering
\begin{tabular}{ll lll}
    \toprule
    Dimension & Metric & Preprocessing time & Space & Query time \\
    \midrule
    $d = 1$ & $L_m$ &
                                          $O(n\log_{\frac{1}{1-\eps}})$
                                            & $O(n\eps^{-1})$
                                                    & $O(\log n + \log\log_{\frac{1}{1-\eps}}n)$ \\
    \midrule
  \multirow{2}{*}{$d = 2$} & $L_\infty$
                       & $\Otilde(n\eps^{-6})^*$ &
                                                   $\Otilde(n\eps^{-6})$
                                                    & $\Otilde(n^{1 /
                                                      2}\eps^{-5/2})^*$ \\
    \cmidrule{2-5}
    & $L_2$           & $\Otilde(n^{1+\delta} + n\eps^{-7})^*$ & $\Otilde(n\eps^{-4})$ & $\Otilde(n^{1 / 2} \eps^{-3/2})^*$ \\
    \bottomrule
\end{tabular}
\caption{
    Our results for the approximate chromatic $k$-nearest neighbors
    problems. Complexities marked with $^*$ are expected bounds.
}
\label{tbl:approximate_solutions}
\end{table}

\section{The one-dimensional problem}
\label{sec:1d}

In this section we consider the case where $P$ is a set of $n$ points
in $\R^1$. In this case all $L_m$-distance metrics with $m \geq 1$ are
the same, i.e. $L_m(a,b)=|a-b|$. We develop a linear space data
structure supporting queries in $O(n^{1/2} \log n)$ time. Building the
data structure will take $O(n^{3/2})$ time. We follow the general
two-step approach sketched in Section~\ref{sec:Introduction}. As we
show in Section~\ref{sub:1d_find_region} there is an optimal, linear-space
data structure with which we can find $\Dk(q)$ in $O(\log n)$
time, even if $k$ is part of the query. As we then briefly describe in
Section~\ref{sub:1d_range_mode} we can directly use the data structure
by Chan~\etal~\cite{chan14linear_space_data_struc_range} to find the
mode color of the points in $\Dk(q)$ in $O(n^{1/2} \log n)$ time.

\subsection{Range finding queries}
\label{sub:1d_find_region}

Given a set $P$ of $n$ points in $\R^{1}$, we wish to store $P$ so that
given a query, consisting of a point $q \in \R^1$ and a natural number
$k$, we can efficiently find the $k^\mathrm{th}$ furthest point from
$q$, and thus $\Dk_m(q)$. We show that by storing $P$ in sorted order
in an appropriate binary search tree, we can answer such queries in
$O(\log n)$ time. To this end, we first consider answering so called
rank queries on a ordered set, represented by a pair of binary search
trees. This turns out to be the crucial ingredient in the data
structure.

\newcommand{\height}{\ensuremath{\mathit{height}\xspace}}
\subparagraph{Rank queries.} Let $R \cup B$ be an ordered set of
elements, and let $T_R$ be a be a binary search tree whose internal
nodes store the elements from $R$, and in which each node is annotated
with the number of elements in that subtree. Similarly, let $T_B$ be a
binary search tree storing $B$. We now show how, given an natural
number $k$, we can find the element among $R \cup B$ with rank $k$
(i.e. the $k^\mathrm{th}$ smallest element), in
$O(\height(T_R) + \height(T_B))$ time, where $\height(T)$ denotes the
height of tree $T$.

Let $r$ be the root of $T_R$, and let $R^<$ and $R^>$ be the subsets
of $R$ with elements smaller and larger than the root,
respectively. We define the root $b$ of $T_B$, and $B^<$ and $B^>$
analogously. Assume without loss of generality that $r > b$, and let
$\ell = |B^<|+|R^<|+1$.
\begin{figure}[tb]
  \centering
  \includegraphics{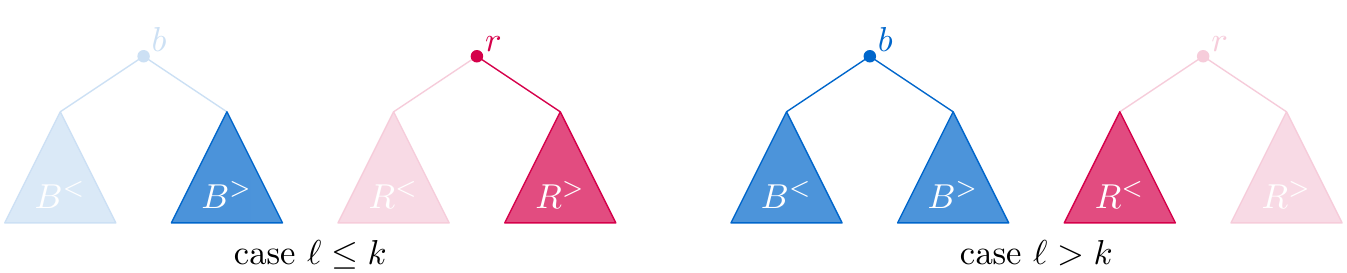}
  \caption{We have $b < r$ so the rank of $r$ is at least
    $\ell = |B^<|+|R^<|+1$. This gives two cases, depending on the
    relation between $\ell$ and $k$. In either case we decrease the
    height of one tree.}
  \label{fig:opt_rank}
\end{figure}
\begin{description}
\item[If $\ell \leq k$] then the item that we search for is not in
  $B^<$, $R^<$, or $\{b\}$. Hence, we can discard $T_{B^<}$,
  $T_{R^<}$, and the blue root $b$, and recursively search for the
  element with rank $k-\ell$ in the set $B^> \cup (R^> \cup
  \{r\})$. We can represent this set using $T_{B^>}$ and the binary
  search tree with root $r$ that has a leaf as its left subtree, and
  $T_{R^>}$ as its right subtree. See Figure~\ref{fig:opt_rank}.
\item[If $\ell > k$] then the item that we search for is not in
  $R^> \cup \{r\}$, since the rank of $r$ is at least $\ell$. Hence,
  we can discard $r$ and $R^>$ and recursively search for the element
  of rank $k$ in $B \cup R^{<}$. This set can be represented by the
  pair of binary search trees $T_B$, $T_{R^<}$.
\end{description}

We continue the search until one of the trees is a leaf. We can then
trivially use the size annotations to find the element of rank $k$ in
the other tree $T$ in $O(\height(T))$ time.

Since each node stores the subtree size, we can compute $\ell$, and
thus determine in which case we are in constant time. In the first
case, we reduce the height of $T_B$ by one. In the second case we
reduce the height of $T_R$ by one. Hence, after at most
$\height(T_B) + \height(T_R)$ rounds, one of the trees is a leaf. Note
that during this process the roles of $R$ and $B$ may switch
(depending on the ordering of the elements stored at their roots), but
this does not affect the running time. We thus obtain:

\begin{lemma}
  \label{lem:rank_queries}
  Let $T_B$ and $T_R$ be two binary search trees with size
  annotations, and let $k$ be a natural number. We can compute the
  element of rank $k$ among $B \cup R$ in
  $O( \height(T_B)+\height(T_R))$ time.
\end{lemma}

\subparagraph{An optimal data structure for computing $\Dk(q)$.} We
store $P$ in a balanced binary search tree $T_P$ with subtree-size
annotations, so that we can: (i) search for the element of rank $k$,
i.e. the $k$ smallest element, and (ii) given a query value
$q \in \R^1$ we can split the tree at $q$ in $O(\log n)$ time. We can
implement $T_P$ using e.g. a red black tree~\cite{tarjan83data}
(although with some care even a simple static balanced binary search
tree will suffice). We will use the split operations only to answer
queries; so we use path copying in this operation, so that we can
still access the original tree once a query
finishes~\cite{sarnak86planar}. The data structure uses $O(n)$ space,
and can be built in $O(n\log n)$ time.

Given a query $(q,k)$ the main idea is now to split $T_P$ into two
trees $T_{P^<}$ and $T_{P^\geq}$, where $P^< \subseteq P$ is the set
of points left of $q$ and $P^\geq = P\setminus P^<$ is the remaining
set of points right of $q$ (or coinciding with $q$). Observe that in
these two trees, the points are actually ordered by distance to $q$
(albeit for $T_{P^<}$ the points are stored in decreasing order while
the points in $T_{P^\geq}$ are stored in increasing order). So, we can
essentially use the procedure from Lemma~\ref{lem:rank_queries} on the
trees $T_{P^<}$ and $T_{P^\geq}$ to find the point in
$P^{\leq} \cup P^{>}$ with rank $k$ (according to the ``by distance to
$q$''-order). The only difference with the algorithm as described
above is that for $T_{P^<}$ the roles of the left and right subtree
are reversed. Splitting $T$ into $T_{P^<}$ and $T_{P^\geq}$ takes
$O(\log n)$ time. Since both subtrees have height at most $O(\log n)$,
Lemma~\ref{lem:rank_queries} also takes $O(\log n)$ time. So, we
obtain the following result:

\begin{theorem}
  \label{thm:1d_finding_query_range}
  Let $P$ be a set of $n$ points in $\R^1$. In $O(n\log n)$ time we
  can build a linear space data structure so that given a query $q,k$
  we can find the smallest disk $\Dk_m(q)$, with respect to any $L_m$ metric,
  containing $\kNN_m(q)$ in $O(\log n)$ time.
\end{theorem}

\subsection{Range mode queries}
\label{sub:1d_range_mode}

What remains is to store $P$ such that given a query interval $Q$ we
can efficiently report the mode color among $P \cap Q$. We use the
following data structure of Chan
\etal~\cite{chan14linear_space_data_struc_range} to this end:

\begin{lemma}[{\cite[Theorem 1]{chan14linear_space_data_struc_range}}]
  \label{lem:array_range_mode_query}
  Let $A$ be an array of size $n$. In $O(n^{3/2})$ time, we can build
  a data structure of size $O(n)$ that reports the mode of a query
  range $A[i \mathrel{:} j]$ in $O(n^{1/2})$ time.
\end{lemma}

By implementing arrays with balanced binary search trees, we can also
implement this structure in the pointer machine model. This increases
the query and preprocessing times by an $O(\log n)$ factor. We then
store (the colors of) the points in increasing order in this
structure. Together with our result from Theorem~\ref{thm:1d_finding_query_range} we obtain:

\begin{theorem}
  \label{thm:1d}
  Let $P$ be a set of $n$ points in $\R^1$. In $O(n^{3/2}\log n)$ time, 
  we can build a data structure of size $O(n)$, that answers chromatic 
  \kNN queries on $P$ in $O(n^{1/2} \log n)$ time.
\end{theorem}

\section{Range finding queries two dimensions}
\label{sec:finding_the_range_2d}

In this section we give a data structure that, given a query point
$q \in \R^2$, reports the smallest range $\Dk_m(q)$ centered at $q$
containing $\kNN_m(q)$. In
Section~\ref{sub:finding_the_k-nearest_neighbors_under_the_L_infty_metric}
we consider the case where the $L_\infty$ metric is used. 
In Section~\ref{sub:finding_the_k-nearest_neighbors_under_the_L_2_metric}
we then consider the $L_2$ metric.

\subsection{The $L_\infty$ metric case}
\label{sub:finding_the_k-nearest_neighbors_under_the_L_infty_metric}

For a given value $r \geq 0$, define
$S(q, r) = [q_x - r, q_x + r] \times [q_y - r, q_y + r]$ as the
axis-aligned square with sidelength $2r$ centered at
$q$.  We call $r$ the \emph{radius} of such a square. Now observe that
$\Dk(q)=S(q,r^*)$ is also an axis-aligned square, in particular with
radius $r^*$ equal to the distance $L_\infty(q,p)$ between $q$ and the
$k^\mathrm{th}$ nearest neighbor $p$ of $q$.
Thus we either have
$r^*=|q_x-p_x|$ or $r^*=|q_y-p_y|$. This leads us to the following
data structure. We store the $x$-coordinates $x_1, \dots, x_n$ of
the points in $P$ in increasing order in a balanced binary search
tree. Similarly, we store the $y$-coordinates of the points in $P$ in
sorted order $y_1, \dots, y_n$.  We set $x_0 = y_0 = -\infty$ and
$x_{n + 1} = y_{n + 1} = \infty$, and call these values coordinates as
well. In addition, we store $P$ in a data structure for
$O(\log n)$ time orthogonal range counting queries, for which we use a range
tree~\cite{DBLP:journals/cacm/Bentley80, willard85orthogonal}. 
The entire data structure can be constructed in
$O(n \log n)$ time, and uses $O(n \log n)$ space.

Now let $x_0, \dots, x_\ell$ be the $x$-coordinates that are at most
$q_x$. The sequence $r_i = |q_x - x_i|$, for $i = 0, \dots, \ell$,
defines a sequence of decreasing radii. See Figure~\ref{fig:L_infty_radii}
for an illustration.
\begin{figure}
  \centering
  \includegraphics{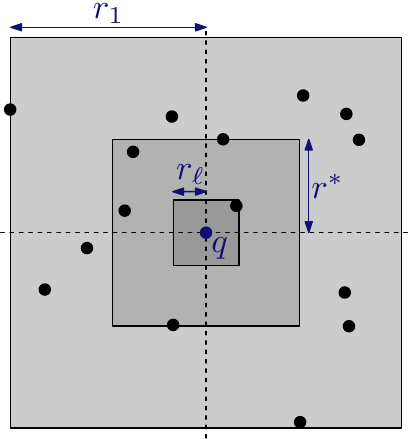}
  \caption{The considered radii $r_i$, and $r^*$, for $k = 5$.}
  \label{fig:L_infty_radii}
\end{figure}
We can find the smallest
radius $r_i$ for which $S(q, r_i)$ contains at least $k$ points by
performing binary search over the radii, performing orthogonal range
counting at each step to guide the search. By performing a similar
procedure for the $x$-coordinates greater than $q_x$, as well as for
the $y$-coordinates, we obtain a set of four squares, that each
contain at least $k$ points.  The smallest of these squares contains
exactly $k$ points and is thus $\Dk(q)$. As each procedure performs
$O(\log n)$ orthogonal range counting queries, we obtain the following
theorem.

\begin{theorem}
  \label{thm:2d_finding_query_range_L_infty_nlogn}
  Let $P$ be a set of $n$ points in $\R^2$.  In $O(n \log n)$ time, we
  can build a data structure of size $O(n \log n)$, that can report
  $\Dk_1(q)$ and $\Dk_\infty(q)$ in $O(\log^2 n)$ time.
\end{theorem}

Following an idea of Chan
\etal~\cite{chan14linear_space_data_struc_range} we can reduce the
space used by replacing the binary range tree used for the orthogonal
range queries by one with fanout $n^\delta$ for some constant $\delta > 0$.

\begin{lemma}
  \label{lem:linear_range_tree}
  Let $P$ be a set of $n$ points in $\R^2$. Let $\delta > 0$ be an
  arbitrarily small constant. In $O(n \log n)$ time,
  we can build a data structure of size $O(n)$, that answers orthogonal
  range reporting and counting queries on $P$ in $O(n^\delta + k')$
  time, where $k'$ is the output size.
\end{lemma}
\begin{proof}
    Fix $\delta > 0$.
    The data structure is a two-dimensional range tree, where
    the nodes in the primary tree have $\lceil n^{\delta / 2} \rceil$
    children, rather than $2$. This gives the primary tree a height of
    $O(2 / \delta)$ instead of $O(\log n)$. The primary tree stores
    the points of $P$ sorted by $x$-coordinate. The internal nodes
    of the tree store balanced binary search trees as associated
    structures, built on the canonical subsets of the internal 
    nodes. These binary search trees store the points sorted by
    $y$-coordinate.

    The preprocessing time is easily seen to be $O(n \log n)$, by
    adapting the construction algorithm for the standard range
    tree to give a primary tree with higher degree. For a single
    level of the primary tree, the union of the canonical subsets
    is equal to $P$. Therefore, the associated structures in all
    nodes of a given level take up $O(n)$ space in total. As there
    are $O(2 / \delta)$ levels in the primary tree, the data structure
    uses $O(2n / \delta) = O(n)$ space in total. The query algorithm
    is similar to that of a range tree. Let $Q$ be the query range,
    and let $Q_\ell$ and $Q_r$ be the $x$-coordinates of the left and 
    right sides of $Q$, respectively. Then in the primary tree, we
    first look for the leftmost and rightmost points in $P$ whose
    $x$-coordinates are between $Q_\ell$ and $Q_r$. This can be
    done by two traversals of the primary tree, from the root
    to the leaves, in $O(2n^{\delta / 2} / \delta)$ time total. Using these
    two search paths, we can now identify $O(2n^{\delta / 2} / \delta)$
    internal nodes, and $O(n^{\delta / 2})$ leaves, such that
    their canonical subsets are disjoint, and together make up
    the set of points in $P$ whose $x$-coordinates are between
    $Q_\ell$ and $Q_r$. By querying the associated structures
    in each of these nodes, taking $O(\log n)$ time each,
    and combining the results, we can answer a query. The total
    query time is therefore $O(n^{\delta / 2} \log n + k') = O(n^\delta + k')$.
\end{proof}

\begin{theorem}
  \label{thm:2d_finding_query_range_L_infty_linear}
  Let $P$ be a set of $n$ points in $\R^2$. Let $\delta > 0$ be an
  arbitrarily small constant. In $O(n \log n)$ time, we
  can build a data structure of $O(n)$ size, that can report $\Dk_1(q)$
  and $\Dk_\infty(q)$ in $O(n^\delta)$ time.
\end{theorem}

\subsection{The $L_2$ metric case}
\label{sub:finding_the_k-nearest_neighbors_under_the_L_2_metric}

In this section, we give two data structures that can report
the disk $\Dk_m(q)$ under the $L_2$ metric. In Section~\ref{subsub:rand_finding_disk}, we give a solution with a randomized
query time bound. Then in Section~\ref{subsub:deter_finding_disk},
we show that with a similar strategy, we can make the query time
hold in the worst case.
Throughout this section, we use $B$ to denote the same
constant used by Agarwal
\etal~\cite{agarwal13semialgebraic}. This is a constant depending
on the number of polynomial inequalities $s$ of the ranges, the
maximum degree $\Delta$ of these polynomials, and the dimension
$d$. In our case. all of these are small constants.

\subsubsection{Randomized $O(n^{1/2} \log^{B+1} n)$ query time}
\label{subsub:rand_finding_disk}

In this section, we give a simple randomized data structure
that finds $\Dk(q)$ in expected $O(n^{1/2} \log^{B+1} n)$ time.
The data structure consists of the large fan-out partition tree
of Agarwal~\etal~\cite{agarwal13semialgebraic}, used for
semialgebraic range searching. Together with this tree, we
take a random sample $P'$ of $P$ by including each point with probability $1/n^{1/2}$. Note that
this random sample will contain $n^{1/2}$ points in expectation.

The main idea is to combine binary search on the ordered 
distances $R = \{L_2(q, p') \mid p' \in P'\}$ with circular range 
counting. Let $r^*$ be the distance between $q$ and its $k^{th}$ 
nearest neighbor among $P$. We then search for two consecutive 
distances $r_i, r_{i+1} \in R$, such that $r_i \leq r^* \leq r_{i+1}$. 
Because the number of points $p \in P$ with $r_i \leq L_2(q, p) \leq r_{i+1}$ is $n^{1/2}$ in expectation, we can then afford
to report these points with semialgebraic range reporting, and
combining binary search with range counting again to find $r^*$.
This leads to an expected query time of $O(n^{1/2} \log^{B+1} n)$.\footnote{We would like to thank an anonymous reviewer for this randomized solution, which led to our current solution in Section~\ref{subsub:deter_finding_disk}.}

\subsubsection{Worst-case $O(n^{1/2} \log^{B+1} n)$ query time}
\label{subsub:deter_finding_disk}

We now show how to achieve the same query time complexity as in 
Section~\ref{subsub:rand_finding_disk}, but in the worst case.

For now, we assume that the set $P$ lies in $D_0$-general 
position for a constant $D_0$ (see~\cite{agarwal13semialgebraic} for a definition). 
The details on this assumption are not important, 
and we will later show how to handle arbitrary point sets.
The data structure consists of two copies of the large fan-out
(fan-out $n^{\delta}$, for some constant $\delta > 0$,)
partition tree of Agarwal~\etal~\cite{agarwal13semialgebraic}, 
built on $P$. The first copy, which we call $\T$, will be 
augmented slightly to support generating candidate ranges 
(disks) that will eventually lead to $\Dk(q)$. The second
copy will be used as a black box, answering circular range
counting queries to guide the search for $\Dk(q)$ by 
counting the number of points inside the candidate ranges.

The tree $\T$ is constructed by recursively partitioning
the space into open, connected regions, called \emph{cells}.
Once a cell contains a small (constant) number of points of 
$P$, the recursion stops and $\T$ gets a leaf node containing
these points. There may be points of $P$ that do not lie
on these cells, but rather on the zero set of the partitioning
polynomial used to partition the space. For range searching
with arbitrary point sets, Agarwal~\etal~\cite{agarwal13semialgebraic}
store these points in an auxiliary data structure. However,
with our assumption that $P$ lies in $D_0$-general position,
we do not need this auxiliary data structure, and will simply
store the points inside a leaf node, whose parent is the node corresponding to the partitioning polynomial.

We augment $\T$ further as follows. We adjust $\T$ such that each 
internal node corresponding to a cell $\omega$ stores an arbitrary point 
in $\omega$. During the construction of $\T$, a point inside
each cell is already computed. Hence the construction time
is unaltered.

\subparagraph*{Answering a query.} To query the structure with a query point $q$, we 
keep track of a set of nodes $N_i$ for each level $i$ of 
$\T$ that is explored by our algorithm.
With slight abuse of terminology, we refer to the sets $P'$
of points stored in the leaves of $\T$ as cells.
Let $\Omega_i$ denote the cells corresponding to the internal
nodes in $N_i$, and $P_i$ denote the cells corresponding to
the leaf nodes in $N_i$.
We maintain the invariant that $p^k(q)$, the $k^{th}$ nearest
neighbor of $q$, is contained in a cell in $\Omega_i \cup P_i$.
Initially, $N_1$ contains the root node. If $\T$ is a single
leaf, then the cell corresponding to the root node will be
stored in $P_1$. Otherwise, it will be stored in $\Omega_1$.

The query algorithm works as follows. Say the algorithm is at level $i$ in $\T$.
For a cell $\omega \in \Omega_i$ stored in an internal node, 
let $p(\omega)$ be the point inside $\omega$ that was stored 
with it. Let
\[
  R_i = \{L_2(q, p(\omega)) \mid \omega \in \Omega_i\} \cup \bigcup_{P'_i \in P_i} \{L_2(q, p) \mid p \in P'_i\}
\]
be the set of distances to these points $p(\omega)$, as well as
to all points stored in the leaves in $P_i$ and in the nodes in $N_i$.
Let $r^*$ be the distance between $q$ and its $k^{th}$ nearest
neighbor among $P$. This is the radius of $\Dk(q)$. To find this radius,
we compute the largest distance $r^- \in R_i$, and the
smallest distance $r^+ \in R_i$, such that
$r^- < r^* \leq r^+$. See Figure~\ref{fig:partition_tree_candidate_ranges}. 
\begin{figure}
  \centering
  \includegraphics[scale=0.75, page=2]{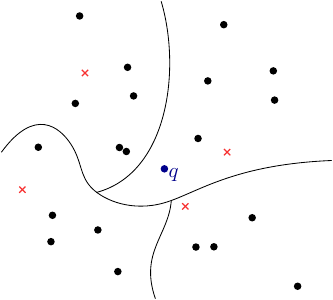}
  \quad
  \includegraphics[scale=0.75, page=3]{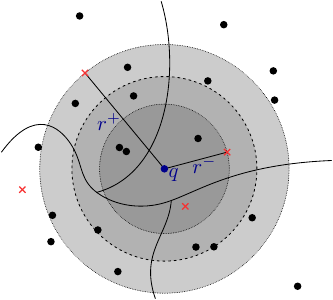}
  \caption{(left) A partitioning of $P$ into four cells.
  The red crosses are the points $p(\omega)$.
  (right) The disk $\Dk(q)$ (dashed boundary) and the disks
  $D(q, r^-)$ (dotted boundary, dark gray) and $D(q, r^+)$
  (dotted boundary, light gray).}
  \label{fig:partition_tree_candidate_ranges}
\end{figure}
If $r^-$ (respectively
$r^+$) does not exist, set it to $0$ (respectively 
$\infty$).
We show how to compute $r^+$. Computing $r^-$ works similarly.

To compute $r^+$, we use a combination of median finding and
binary search. For some radius $r \in R_i$, we decide if
$r^+ > r$ or $r^+ \leq r$ using a circular counting query. If 
$D(q, r)$ contains less than $k$ points, we have $r^+ > r$. 
Otherwise, we have $r^+ \leq r$. By performing this procedure for 
the median radius in $R_i$, we can discard half of $R_i$ with
each query.

Once we have that $r^+$ is the distance between $q$ and a point
in $P$ (it is constructed through a leaf node of $\T$), we 
have found $r^*$. We then terminate the algorithm, returning
the disk with the found radius. Otherwise we continue the 
search in the next level of $\T$.
To continue the search through the tree, we construct the set
$N_{i+1}$ by replacing every node 
$\nu \in N_i$ with its children whose cells are crossed by
one of $D(q, r^-)$ and $D(q, r^+)$.
A cell $\omega$ is \emph{crossed} by a disk $D$ if 
$\omega \cap D \neq \emptyset$ and $\omega \nsubseteq D$.
The sets $\Omega_{i+1}$ and $P_{i+1}$ are then constructed
from these child nodes. Once these sets are constructed, we 
advance the algorithm to level $i + 1$ and repeat the procedure.

We prove the correctness of the query algorithm, and bound
its time complexity, in the following lemmas.

\begin{lemma}
  The query algorithm correctly returns $\Dk(q)$.
\end{lemma}
\begin{proof}
  We claim that this algorithm correctly returns $\Dk(q)$.
  First, note that if the algorithm returns a disk, that disk 
  contains $k$ points, and its radius is equal to the distance
  between $q$ and a point of $P$.
  Therefore, it is indeed $\Dk(q)$. We now show that our 
  algorithm will always return a disk, 
  and therefore that our algorithm is correct. To show this, it 
  suffices to show that for every level $i$ of $\T$ traversed
  by our algorithm, the set $\Omega_i \cup P_i$ contains the 
  cell containing $p^k(q)$.

  We give a proof by induction.
  It holds trivially that $\Omega_1 \cup P_1$ contains a cell
  containing $p^k(q)$. We prove that if $\Omega_i \cup P_i$ 
  contains a cell $\omega'$ containing $p^k(q)$, 
  then either the algorithm terminates and returns $\Dk(q)$, or
  there is a cell $\omega \in \Omega_{i+1} \cup P_{i+1}$ 
  containing $p^k(q)$.
  
  If $\omega' \in P_i$, then $R_i$ will contain $r^*$. It will 
  then find $r^+ = r^*$ and terminate, returning $D(q, r^+) = \Dk(q)$.
  Now assume that $\omega' \in \Omega_i$.
  Let $\nu \in N_i$ be the internal node corresponding to
  $\omega'$. Now let $\omega$ be the cell stored in a child 
  of $\nu$, such that $p^k(q)$ lies in $\omega$. We show that 
  $\omega \in \Omega_{i+1} \cup P_{i+1}$.

  Our algorithm performs repeated median finding on the radii in 
  $R_i$, resulting in the largest radius $r^- \in R_i$ and
  smallest radius $r^+ \in R_i$, such that $D(q, r^-)$,
  respectively $D(q, r^+)$, contains less than, respectively 
  at least, $k$ points of $P$. Let $r_{\omega}$ be the
  distance between $q$ and $p(\omega)$, the point in $\omega$ that was
  stored in $\T$. If $r_{\omega} < r^*$, then we have
  that $r_{\omega} \leq r^-$, implying that $D(q, r^-)$
  intersects $\omega$. Also, because $\omega$ is open, and because $p^k(q) \in \omega$, there must be a point $p' \in \omega$ such that $r^- < r^* = L_2(q, p^k(q)) < L_2(q, p')$. This shows that $\omega$ is not contained in $D(q, r^-)$, and thus that $D(q, r^-)$ crosses $\omega$.
  With similar reasoning, it can be seen that if $r_{\omega} \geq r^*$, then $D(q, r^+)$ crosses $\omega$.
  Thus, $\omega$ will be crossed by at least one of $D(q, r^-)$ and $D(q, r^+)$, and thus $\omega \in \Omega_{i+1} \cup P_{i+1}$.
This proves the correctness of our algorithm.
\end{proof}

\begin{lemma}
  The running time of the query algorithm is $O(n^{1/2} \log^{B+1} n)$.
\end{lemma}
\begin{proof}
  In every level of the partition tree, we perform repeated 
  median finding on the distances stored in the sets $R_i$, 
  combined with circular range counting at every step. In 
  level $i$, the time spent performing repeated median 
  finding is $O(|R_i|)$, and we perform $O(\log |R_i|)$ 
  circular range counting queries. As circular range 
  counting takes $O(n^{1/2} \log^B n)$ time with the 
  partition tree~\cite{agarwal13semialgebraic}, this brings 
  the total time spent in level $i$ to 
  $O(|R_i| + n^{1/2} \log^{B+1} |R_i|)$.
  Computing all cells $\omega$ in children of nodes in $N_i$ then takes
  $r^c$ time, where $r = n^\delta$ for a small constant $\delta > 0$ is 
  a parameter determining the fan-out of the partition tree, and $c$ is a constant.
  Because the partition tree has constant height, the
  total query time is $O(\sum_i |R_i| + n^{1/2} \log^{B+1} n + r^c)$.
  By setting $\delta \leq 1 / (2c)$ the query time becomes $O(\sum_i |R_i| + n^{1/2} \log^{B+1} n)$.
  We argue that $\sum_i |R_i| = O(n^{1/2} \log^B n)$.

  Let $Q(n)$ be the number of nodes of $\T$ expanded by the
  query algorithm. Note that $\sum_i |R_i| = O(Q(n))$.
  Following the proof of Agarwal~\etal~\cite{agarwal13semialgebraic} on the query time of the large
  fan-out partition tree, we assume that 
  the number of nodes $Q(n)$ we expand satisfies $Q(n') \leq {n'}^{1/2} \log^B n'$,
  for all $n_0 < n' < n$, where
  $n_0$ is a suitable constant determining how many
  points of $P$ are stored in each leaf node of $\T$.
  We now show that $Q(n) = O(n^{1/2} \log^B n)$.
  
  If $n \leq n_0$, we have that $Q(n) = O(n) = O(1)$, and thus
  the bound of $Q(n) = O(n^{1/2} \log^B n)$ holds. We therefore
  assume $n > n_0$.
  Agarwal~\etal~\cite[Lemma 4.3]{agarwal13semialgebraic} 
  proved that for any node $\nu$ of $\T$, and any disk $D$,
  the number of children of $\nu$ whose stored cell is 
  crossed by $D$ is at most $C r^{1/2}$, for a constant $C$ independent
  of $r$. The cells stored in children of $\nu$ will contain 
  at most $n / r$ points. This means that $Q(n)$ satisfies the 
  recurrence $Q(n) \leq C r^{1/2} Q(n/r)$. By our assumption 
  on a bound for $Q(n')$, for $n' < n$, we get that
  $Q(n) \leq C r^{1/2} (n/r)^{1/2} \log^B (n/r)$.
  Since we chose $r = n^\delta$ for a small constant $\delta > 0$, this bound simplifies to
  $Q(n) \leq C (1-\delta)^B n^{1/2} \log^B n = O(n^{1/2} \log^B n)$. The bound of $O(n^{1/2} \log^{B+1} n)$ on the
  query time now follows.
\end{proof}

With the preprocessing time and space taken from Agarwal~\etal~\cite{agarwal13semialgebraic}, we obtain the following 
result.

\begin{lemma}
  Let $P$ be a set of $n$ points in $\R^2$, in $D_0$-general position for a constant $D_0$. Let $\delta > 0$ be an arbitrarily small constant. In $O(n^{1 + \delta})$ expected time, we can
  build a data structure of $O(n)$ size, that can report $\Dk_2(q)$ in
  $O(n^{1/2} \log^{B+1} n)$ time.
\end{lemma}

\subparagraph*{Handling arbitrary point sets.}
We can lift the assumption that $P$ lies in $D_0$-general 
position by applying the perturbation scheme of Yap~\cite{yap90perturbation}.
We refer to~\cite{agarwal13semialgebraic} for details, but
using the perturbation scheme, $\T$ can be constructed on a 
set of points $P'$ in $D_0$-general position, obtained by
perturbing the individual points in $P$ by infinitesimal amounts. Clearly, using the
new tree for our query algorithm can result in a disk $D$ 
whose radius is infinitesimally smaller or larger than that of $\Dk(q)$.
In particular, the query algorithm returns the disk that has
the perturbed version of $p^k(q)$ on its boundary.
We can easily augment the data structure to return a disk with
$p^k(q)$ itself on the boundary, rather than the perturbed 
version.

We augment the data structure to not only store the perturbed points
in the leaves of $\T$, but also their original counterparts.
Because our query
algorithm searches for the point ${p'}^k(q)$, the point $p^k(q)$
after perturbation, we can then augment the query algorithm 
to return the original point stored alongside the perturbed 
one. This returned point will be $p^k(q)$. With this point, we
can then construct $\Dk(q)$. This gives us the following.

\begin{theorem}
  \label{thm:2d_finding_query_range_L2}
  Let $P$ be a set of $n$ points in $\R^2$. Let $\delta > 0$ be an arbitrarily small constant. In $O(n^{1 + \delta})$ expected time, we can
  build a data structure of $O(n)$ size, that can report $\Dk_2(q)$ in
  $O(n^{1/2} \log^{B+1} n)$ time.
\end{theorem}

\section{Range mode queries in two dimensions}
\label{sec:range_mode_2d}

In this section we discuss answering range mode queries in $\R^2$, and
see how we can use them together with the data structures from
Section~\ref{sec:finding_the_range_2d} to answer chromatic \kNN
queries. In Section~\ref{sec:rm2_L2} we review the data structure of
Chan~\etal~\cite{chan14linear_space_data_struc_range} that can answer
range mode queries with disks (i.e. in the $L_2$ metric). We then show
in Section~\ref{sec:rm2_Linfty} how this approach can answer range mode
queries with squares (disks in the $L_\infty$ metric). This leads to
better query times compared to directly using the existing range mode
data structures for orthogonal ranges. Finally, in
Section~\ref{sec:rm2_prep} we show that these data structures can be
built efficiently, and show that this leads to efficient data
structures for chromatic \kNN queries.

\subsection{The $L_2$ metric case} 
\label{sec:rm2_L2}

In this section, we present the data structure of Chan~\etal~\cite{chan14linear_space_data_struc_range} for finding the mode color among points in query disks.
This data structure can then be used in conjunction with that of Section~\ref{sub:finding_the_k-nearest_neighbors_under_the_L_2_metric} to get a data structure for chromatic $k$-nearest neighbors queries.

We first transform the problem, by performing a lifting map and
subsequently dualizing the points. For the lifting map, we map 
each point $p \in P$ to the point
$\hat{p} = (p_x, p_y, p_x^2 + p_y^2)$, essentially lifting the points
of $P$ to the three-dimensional unit paraboloid. Call the set of
lifted points $\hat{P}$. In the following lemma, we will show 
that a disk $D = D(q, r)$ corresponds to the halfspace
$h^-(q, r) : z \leq 2q_x x + 2q_y y - q_x^2 - q_y^2 + r^2$.
See Figure~\ref{fig:lifting_transformation} for an illustration
of this fact.
\begin{figure}
    \centering
    \includegraphics[scale=0.75]{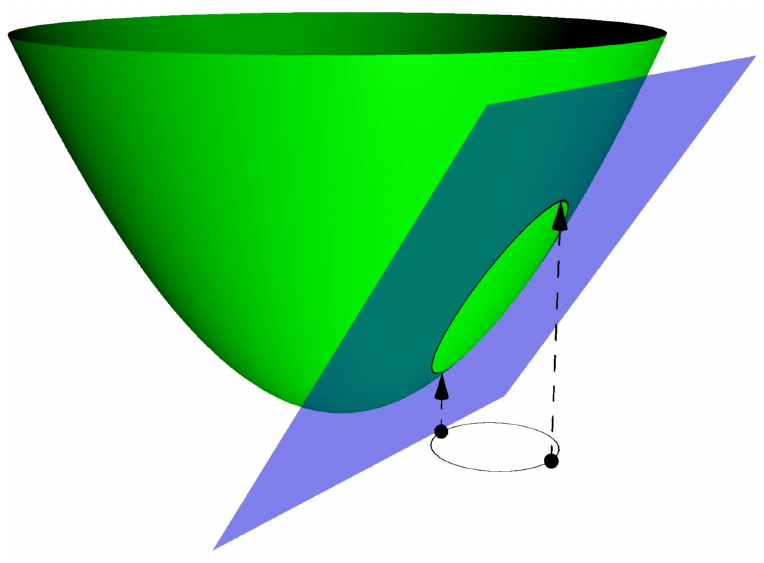}
    \caption{The lifting transformation used to transform a disk 
    $D(q, r)$ to the halfspace $h^-(q, r)$.}
    \label{fig:lifting_transformation}
\end{figure}

\begin{lemma}
  \label{lem:lifting_map}
    Let $D = D(q, r)$ be a disk in $\R^2$.
    A point $p \in \R^2$ lies in $D$ if and only if 
    $\hat{p} \in \R^3$ lies in $H(q, r)$.
\end{lemma}
\begin{proof}
    Let $p \in D$ be a point. Then we have that 
    $(p_x - q_x)^2 + (p_y - q_y)^2 \leq r^2$. This implies that
    $p_x^2 + p_y^2 \leq 2q_x p_x + 2q_y p_y - q_x^2 - q_y^2 + r^2$,
    from which it can be seen that $\hat{p} \in H(q, r)$.
    Following the reasoning backwards completes the proof.
\end{proof}

Let $h^*$ be the point dual to the plane $h$ bounding $h^-(q, r)$.
By taking the dual of $\hat{P}$, a set of planes $H$ is obtained, 
such that the points inside $D$ correspond to the planes below
$h^*$.
We will build a data structure on the set of planes $H$, that can find the mode color among these planes.

The data structure is based on the following observation by Krizanc~\etal~\cite{DBLP:journals/njc/KrizancMS05}:

\begin{lemma}[Krizanc~\etal~\cite{DBLP:journals/njc/KrizancMS05}]
    Let $A$ and $B$ be multisets.
    If $c$ is a mode of $A \cup B$, and $c \notin A$, then $c$ is a mode of $B$.
\end{lemma}

We apply this observation to $\frac{1}{r}$-cuttings.
A \emph{$\frac{1}{r}$-cutting} $\Xi$ of $\A(H)$ is a set of simplices with
disjoint interiors, that together cover $\R^3$, such that the interior of
a simplex in $\Xi$ is intersected by at most $n / r$ planes in $H$.
Now let $q \in \R^3$
be a query point.  Given a $\frac{1}{r}$-cutting $\Xi$ of $\A(H)$, let
$\Delta \in \Xi$ be a simplex containing $q$, for which there is a
point in $\Delta$ strictly above $q$.  The mode color among the planes
below $q$ is then either the color of a plane in the conflict list
$H_\Delta$ of $\Delta$ (multiset $A$), or it is the mode color among
the planes below $\Delta$ (multiset $B$). Note that Chan
\etal~\cite{chan14linear_space_data_struc_range} use slightly
different sets. Our choice makes extending the result to the
$L_\infty$ metric (as we do in Section~\ref{sec:rm2_Linfty}) slightly easier.

\subparagraph*{The data structure.}
For the data structure, we now create a $\frac{1}{r}$-cutting $\Xi$ on $H$, and for every 
simplex $\Delta \in \Xi$, we compute and store the mode color among the planes below $\Delta$.
Aside from this cutting, we construct a point-location data structure for quickly finding the 
simplex containing a query point. We also need the conflict lists of the simplices.
However, as they have a total size of $O(nr^2)$, we will use a seperate data structure
that computes the conflict list during a query. Lemma~\ref{lem:cuttings_and_aux_data_structures} states our results on cuttings and 
computing their conflict lists. The following lemma is used to speed up the 
computation of the conflict lists for our problem, where the planes of $H$ 
originate from two-dimensional points.

\begin{lemma}
  \label{lem:level_queries_2d}
  Let $a \in R^3$ be a point. We can construct, in constant time, a disk
  $D_a$ in $\R^2$, such that the points of $P$ inside $D_a$ correspond to 
  those planes of $H$ below $a$.
\end{lemma}
\begin{proof}
  Let $m = (-a_x / 2, -a_y / 2)$ and $r = \sqrt{a_z - \frac{a_x^2 + a_y^2}{4}}$.
  The disk $D_a = D(m, r)$ with center $m$ and radius $r$ is then lifted to the
  halfspace
  \[
    h^-(m, r) : z \leq -a_x - a_y + \frac{a_x^2 + a_y^2}{4} + a_z - \frac{a_x^2 + a_y^2}{4} = -a_x - a_y + a_z.
  \]
  The plane bounding this halfspace is dual to the point $(a_x, a_y, a_z) = a$.
  Thus, the points of $P$ inside $D_a$ correspond to the planes of $H$
  below $a$.
\end{proof}

\begin{lemma}
  \label{lem:cuttings_and_aux_data_structures}
  Let $r \in [1, n]$ a parameter. We can store
  a $\frac{1}{r}$-cutting of $\A(H)$ of $O(r^3)$ size in a data structure of $O(n + r^3)$ size,
  so that the simplex $\Delta \in \Xi$ containing a query point can be found in 
  $O(\log r)$ time, and so that reporting the conflict list of $\Delta$ takes 
  $O(n^{1/2} \polylog n + n / r)$ time. Building the data
  structure takes expected $O(n^{1 + \delta} + nr^2)$ time.
\end{lemma}
\begin{proof}
  Using the result of Chazelle~\cite{DBLP:journals/dcg/Chazelle93a}, we can construct, 
  in $O(nr^2)$ time, a $\frac{1}{r}$-cutting $\Xi$ of $O(r^3)$ size, along with a
  point-location data structure, also of $O(r^3)$ size, that answers queries in 
  $O(\log r)$ time. 
  
  We now present a data structure for reporting the planes intersecting a simplex $\Delta$, following the ideas of Chan~\etal~\cite{chan14linear_space_data_struc_range}.
  Any plane intersecting the interior of $\Delta$ intersects the interior of an edge of $\Delta$ (without containing the edge).
  Beause $\Delta$ has only $O(1)$ edges, we will report the planes 
  intersecting an edge $e$ of $\Delta$ for all edges of $\Delta$ 
  individually. As we will then show, the number of planes reported 
  this way will still be $O(n / r)$, even though we might report a 
  plane multiple times.

  Let $\overline{ab}$ be a line segment in $\R^3$, of which we want to
  report the planes intersecting it. Note that a plane intersects 
  the interior of $\overline{ab}$
  if and only if it is either strictly above $a$ and strictly below $b$, or 
  strictly below $a$ and strictly above $b$. The first set of planes (and 
  similarly, the second set,) can be reported using two-dimensional
  semialgebraic range reporting, using Lemma~\ref{lem:level_queries_2d}.
  Let $D_a$ and $D_b$ be the two disks in $\R^2$, such that the points inside
  the disks correspond to the planes below $a$ and $b$, respectively.
  Using the result of Agarwal~\etal~\cite{agarwal13semialgebraic},
  we can construct a linear-size data structure in $O(n^{1+\delta})$
  expected time, which reports the $k'$ points strictly outside $D_a$ and strictly 
  inside $D_b$, and thus the $k'$ planes intersecting the interior of $\overline{ab}$,
  in $O(n^{1/2} \polylog n + k')$ time.

  Performing two of the above reporting queries per edge of $\Delta$,
  we obtain a query time of $O(n^{1/2} \polylog n + k')$, where $k'$
  is the number of planes (including duplicates) reported in total.
  Note that we only report planes that intersect the interior of an edge
  of $\Delta$, and thus we only report planes in the conflict list of
  $\Delta$. Furthermore, because we perform only a constant number of
  reporting queries, $k'$ will be at most a constant factor greater than
  the size of $H_\Delta$. This gives a total query time of $O(n^{1/2} \polylog n + n / r)$, proving the lemma.
\end{proof}

Finally, we build a number of data structures for level queries.
One is build on each set of planes in $H$ that share the same color.
For these data structures, we use Lemma~\ref{lem:level_queries_2d} to handle
the queries with two-dimensional semialgebraic range counting.
We use the linear-space solution of Agarwal~\etal~\cite{agarwal13semialgebraic},
which can be constructed in $O(n^{1+\delta})$ expected time and answers 
queries in $O(n^{1/2} \polylog n)$ time.

We now show how to compute the colors stored in the simplices. This can be done in
$O(nr^3)$ time, after we created pointers from each plane to a counter counting
the frequency of the plane's color. The color of a simplex can then be computed
by scanning through all planes, upping a plane's color's frequency if the plane
lies below the simplex, and keeping track of the color with the highest
frequency. This takes $O(n)$ time per simplex, totalling $O(nr^3)$ time.
This brings the total size of the data structure to $O(n + r^3)$, and the total expected preprocessing time to $O(n^{1 + \delta} + nr^3)$.

\subparagraph*{Answering a query.}
Querying the structure with a point $q$ is done by first finding a simplex $\Delta \in \Xi$ containing $q$, for which there is a point in $\Delta$ strictly above $q$.
This can be done in $O(\log r)$ time.
Then, the conflict list of $\Delta$ is reported, taking $O(n^{1/2} \polylog n + n / r)$ time.
Finally, for each color $c$ of the $O(n / r)$ reported planes, as well as for the color stored in $\Delta$, the number of planes below $q$ that have color $c$ is counted.
Recall that this can be done using level queries.
The mode color among the planes below $q$ is the color with the highest count.

Let $H_c$ be the set of planes with color $c$, and let $C_\Delta$ be the set of colors for which we perform a level query.
Note that $|C_\Delta| = O(n / r)$.
H\"{o}lders inequality now states that
\[
    \left( \sum_{i = 1}^{m} x_i^a y_i^b \right)^{a + b} \leq \left( \sum_{i = 1}^{m} x_i^{a + b} \right)^a \left( \sum_{i = 1}^{m} y_i^{a + b} \right)^b
\]
for any two sequences $x_1, \dots, x_m$ and $y_1, \dots, y_m$ and any two positive values $a$ and $b$.
Applied to the level queries, we find that the level queries take a total of
\begin{align*}
    O\left( \sum_{c \in C_\Delta} |H_c|^{1 / 2} \polylog |H_c| \right) 
    &= O\left( \left[ \sum_{c \in C_\Delta} 1 \right]^{1 / 2} \left[ \sum_{c \in C_\Delta} |H_c| \right]^{1 / 2} \polylog n \right) \\
    &= O((n / r^{1 / 2}) \polylog n)
\end{align*}
time.  In total, the query time of the data structure is thus
$O((n / r^{1 / 2}) \polylog n)$. We thus get the following result
(essentially Theorem 17 of Chan
\etal~\cite{chan14linear_space_data_struc_range} in the case of halfspaces
in $\R^3$, though slightly improved for the case of two-dimensional circular ranges):

\begin{proposition}
  \label{prop:range_modeL2}
  Let $P$ be a set of $n$ colored points in $\R^2$ and $r \in [1, n]$
  a parameter. In $O(n^{1 + \delta} + nr^3)$ expected time, we can build a
  data structure of $O(n + r^3)$ size, that reports the mode color
  among the points in a query disk in $O((n / r^{1 / 2}) \polylog n)$
  time.
\end{proposition}

\subsection{The $L_\infty$ metric case} 
\label{sec:rm2_Linfty}

We can use the data structure from Section~\ref{sec:rm2_Linfty} for the
problem under the $L_\infty$ metric as well.
For a two-dimensional point $p$, the graph of its distance function
\[
    L_\infty(p, q) = \max\{|p_x - q_x|, |p_y - q_y|\}
\]
forms an upside-down pyramid $\nabla_p$ in $\R^3$.  If we now have a
two-dimensional point $q$, then we have that $L_\infty(p, q) \leq r$
if and only if $(q_x, q_y, r)$ lies above $\nabla_p$.  See
Figure~\ref{fig:L_infty_graph} for an illustration.  Thus, by looking
at the graphs of the distance functions for each point in $P$, the
problem is transformed to that of finding the mode color among those
graphs in $\R^3$ that lie below a given query point. We now alter
the data structure of Section~\ref{sec:rm2_L2} to work for these
graphs.

\begin{figure}
\centering
\includegraphics[width=0.4\textwidth]{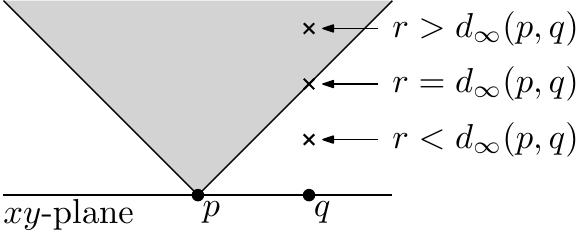}
\caption{
    An illustration of the upside-down pyramid $\nabla_p$.
    The distance between $p$ and $q$ under the $L_\infty$ metric (in $\R^2$) determines whether the point $(q_x, q_y, r)$ lies strictly under, on, or strictly above $\nabla_p$.
}
\label{fig:L_infty_graph}
\end{figure}

\subparagraph*{The data structure.}
We write $\nabla_p$ to denote the graph $L_\infty(p, \cdot)$ for a point $p \in \R^2$, and write $\nabla = \{\nabla_p \ | \ p \in P\}$ to be the set of graphs for all points in $P$.
From now on, we will refer to the graphs as \emph{pyramids}.
Like for planes, we define a $\frac{1}{r}$-cutting $\Xi$ of $\A(\nabla)$ is a subdivision of $\R^3$ 
into (possibly unbounded) simplices with disjoint interiors such that the interior of each 
simplex is intersected by at most $n / r$ pyramids of $\nabla$. The set of pyramids
intersecting the interior of a simplex $\Delta \in \Xi$ is again called the conflict list
of $\Delta$. For the data structure, we need to construct a $\frac{1}{r}$-cutting for the arrangement $\A(\nabla)$ formed by $\nabla$, together with a point-location data structure. (See Lemma~\ref{lem:cuttings_and_aux_data_structures} for the case of planes.) 
The following lemma shows that we can use cuttings for planes to construct the desired cuttings.

\begin{lemma}
    Let $\nabla$ be a set of $n$ pyramids and $r \in [1, n]$ a parameter.
    In $O(nr^2)$ time, we can construct $\frac{1}{r}$-cutting of $\A(\nabla)$ that consists of $O(r^3)$ simplices.
    With $O(r^3)$ additional preprocessing and space, a data structure for point location in the cutting can be build which answers queries in $O(\log r)$ time.
\end{lemma}
\begin{proof}
    Let $\nabla_p$ be a pyramid in $\nabla$.
    This pyramid is contained in the union of the four planes
    \begin{align*}
        p_x = x \pm z, \\
        p_y = y \pm z.
    \end{align*}
    Now let $H$ be the set of all $4n$ (not necessarily distinct) planes whose union contains all of $\nabla$.
    Let $\Xi$ be a $\frac{1}{4r}$-cutting of $\A(H)$.
    If the interior of a simplex $\Delta \in \Xi$ intersects a pyramid in $\nabla$, it must also intersect a plane in $H$.
    Therefore, the number of pyramids in $\nabla$ that intersect the interior of $\Delta$ is bounded by the number of planes in $H$ that intersect the interior of $\Delta$, which is at most $4n / (4r) = n / r$.
    The cutting $\Xi$ is therefore a $\frac{1}{r}$-cutting for $\A(\nabla)$.

    If the planes of $H$ are in general position, we can use a result by Chazelle~\cite{DBLP:journals/dcg/Chazelle93a} to finish the proof.
    Unfortunately, the planes are definitely not in general position.
    For one, the planes constructed for a single pyramid $\nabla_p$ all contain $p$, meaning that four planes intersect in a single point.
    Moreover, for any two points $p_1$ and $p_2$, we have that the planes constructed for $\nabla_{p_1}$ and $\nabla_{p_2}$ form four pairs of parallel (or even coinciding) planes.
    Luckily, we can perturb the set of planes $H$ such that they do lie in general position, while keeping the properties of the constructed cutting intact~\cite{DBLP:journals/siamcomp/EmirisC95, DBLP:journals/tog/EdelsbrunnerM90}.
    That is, the $\frac{1}{4r}$-cutting constructed for the perturbed planes will be a $\frac{1}{r}$-cutting for $\A(\nabla)$.
    The lemma now follow from~\cite{DBLP:journals/dcg/Chazelle93a}.
\end{proof}

Like for the case of planes, we need the conflict lists of the simplices. 
Again, however, they have a total size of $O(nr^2)$. Hence, we use the same approach as for 
planes, creating an auxiliary data structure that can report the conflict list of a simplex. 
Note, though, that we only report a specific subset of the conflict list. This subset
is still enough for the queries, and is easier to report than the entire conflict lists.
The data structure is given in Lemma~\ref{lem:reporting_conflict_list_pyramids}.
We first prove the following lemma.

\begin{lemma}
  \label{lem:simplex_containment_conditions}
    Let $p \in \R^2$ be a point.
    A simplex $\Delta \subset \R^3$ lies above $\nabla_p$ if and only if all vertices of $\Delta$ lie above $\nabla_p$, and any unbounded edge of $\Delta$, when seen as a ray originating from some vertex of $\Delta$, has direction $\vec{d}$, with
    \[
        \left\{
        \begin{matrix}
            -\vec{d}_z \leq \vec{d}_x \leq \vec{d}_z, \\
            -\vec{d}_z \leq \vec{d}_y \leq \vec{d}_z.
        \end{matrix}
        \right.
    \]
\end{lemma}
\begin{proof}
    Let $p \in \R^2$ be a point and let $\Delta \subset \R^3$ be a simplex.
    Because $\Delta$ is a convex object, and because the area above $\nabla_p$ forms a convex region, we have that $\Delta$ lies above $\nabla_p$ if and only if all vertices and edges of $\Delta$ lie above $\nabla_p$.
    By the same argument, an edge of $\Delta$ lies above $\nabla_p$ if and only if its incident vertices lie above $\nabla_p$.
    This proves the theorem for bounded simplices.

    Assume that $\Delta$ is unbounded, and that all vertices of $\Delta$ lie above $\nabla_p$.
    Let $e$ be an unbounded edge of $\Delta$, incident to vertex $v$.
    Let $\vec{d}$ be the direction of $e$, when seen as a ray originating from $v$.
    A point $(x, y, z)$ on $e$ now satisfies
    \[
        \begin{cases}
            x = v_x + \lambda \vec{d}_x \\
            y = v_y + \lambda \vec{d}_y \\
            z = v_z + \lambda \vec{d}_z
        \end{cases}
    \]
    for some $\lambda \geq 0$.

    A point $(x, y, z)$ lies above $\nabla_p$ if and only if $L_\infty(p, (x, y)) \leq z$.
    It follows that $e$ lies above $\nabla_p$ if and only if
    \[
        \max\left\{ | p_x - (v_x + \lambda \vec{d}_x) |, | p_y - (v_y + \lambda \vec{d}_y) | \right\} \leq v_z + \lambda \vec{d}_z
    \]
    for all $\lambda \geq 0$, and therefore if and only if the system of equations
    \[
        \begin{cases}
            v_x + v_z + \lambda (\vec{d}_x + \vec{d}_z) \geq p_x \\
            v_x - v_z + \lambda (\vec{d}_x - \vec{d}_z) \leq p_x \\
            v_y + v_z + \lambda (\vec{d}_y + \vec{d}_z) \geq p_y \\
            v_y - v_z + \lambda (\vec{d}_y - \vec{d}_z) \leq p_y
        \end{cases}
    \]
    is satisfied for all $\lambda \geq 0$.
    Using our assumption that all vertices of $\Delta$, and $v$ in particular, lie above $\nabla_p$, we now have that $e$ lies in $\nabla$ if and only if the system of equations
    \[
        \begin{cases}
            \lambda (\vec{d}_x + \vec{d}_z) \geq 0 \\
            \lambda (\vec{d}_x - \vec{d}_z) \leq 0 \\
            \lambda (\vec{d}_y + \vec{d}_z) \geq 0 \\
            \lambda (\vec{d}_y - \vec{d}_z) \leq 0
        \end{cases}
    \]
    is satisfied for all $\lambda \geq 0$.
    This implies that $e$ is contained in $\nabla_p$ if and only if $-\vec{d}_z \leq \vec{d}_x \leq \vec{d}_z$ and $-\vec{d}_z \leq \vec{d}_y \leq \vec{d}_z$.
\end{proof}

\begin{lemma}
  \label{lem:reporting_conflict_list_pyramids}
    Let $\delta > 0$ be an arbitrarily small constant.
    We can store the cutting $\Xi$ in a data structure of $O(n + r^3)$ size, so that given
    a simplex $\Delta \in \Xi$ and query point $q$, the subset of the conflict list of 
    $\Delta$, containing those pyramids below $q$, can be reported in $O(n^\delta + n / r)$ 
    time.
    Alternatively, with $O(n \log n + r^3)$ space, the query time can be reduced to $O(\log n + n / r)$.
\end{lemma}
\begin{proof}
    Let $\Delta \subset \R^3$ be a simplex and let $q \in \Delta$ be a point.
    The pyramid $\nabla_p$ lies under $q$ if and only if $L_\infty(p, (q_x, q_y)) \leq q_z$.
    That is, $\nabla_p$ lies under $q$ if and only if $p \in [q_x - q_z, q_x + q_z] \times [q_y - q_z, q_y + q_z] = S_q$.
    See also Figure~\ref{fig:L_infty_graph}.
    Now, using Lemma~\ref{lem:simplex_containment_conditions}, we can check in $O(1)$ time whether $\Delta$ does not lie above $\nabla_p$, or whether it lies above $\nabla_p$ if and only if the vertices of $\Delta$ lie above $\nabla_p$.
    Without loss of generality, assume that $\Delta$ lies above $\nabla_p$ if and only if all its vertices lie above $\nabla_p$, and assume that $\Delta$ is a bounded simplex.

    Like for the query point $q$, a vertex $v$ of $\Delta$ lies above $\nabla_p$ if and only if $p$ lies in some two-dimensional axis-aligned rectangle $S_v$.
    Therefore, the points $p$ that we have to report are those in the region
    \[
        R = S_q \setminus \left( \bigcap_{v \in V_\Delta} S_v \right),
    \]
    where $V_\Delta$ is the set of vertices of $\Delta$.
    Because the intersection of two axis-aligned rectangles is an axis-aligned rectangle, the region $R$ is the difference of two axis-aligned rectangles.
    As there are only $O(1)$ vertices and edges in the union of the rectangles, their vertical decomposition will consist of $O(1)$ axis-aligned rectangles as well.
    By performing orthogonal range reporting on these rectangles, making sure not to include edges of the vertical decomposition in multiple queries, we can now report the points $p \in P$, and therefore the pyramids $\nabla_p$ in the conflict list of $\Delta$ that lie below $q$, using $O(1)$ orthogonal range reporting queries.
    Using a range tree (either using $O(n \log n)$ space~\cite{willard85orthogonal}, 
    or $O(n)$ space with Lemma~\ref{lem:linear_range_tree}) for these queries results in the bounds in the lemma.
\end{proof}

The final part of the data structure is the set of range counting data structures for the different sets of pyramids with equal colors.
For each color $c$, let $\nabla_c \subset \nabla$ be the set of pyramids with color $c$.
Because a point $q$ lies above a pyramid $\nabla_p$ if and only if $L_\infty(p, (q_x, q_y)) \leq q_z$, counting the number of pyramids below $q$ that have color $c$ is equivalent to counting the number of points of $P$ in the range $S((q_x, q_y), q_z) = [q_x - q_z, q_x + q_z] \times [q_y - q_z, q_y + q_z]$ that have color $c$.
The range counting data structures therefore take the form of orthogonal range counting data structures.

Let $P_c \subset P$ be the set of points that have color $c$.
Using a range tree, the total preprocessing time for these data structures is
$O(n \log n)$.
Depending on whether we use a range tree of $O(n)$ size or $O(n \log n)$ size, the total space for the range counting data structures is $O(n)$ or $O(n \log n)$, respectively.

\subparagraph*{Answering a query.}
Querying the structure with a point $q$ works similarly to querying the structure of
Section~\ref{sec:rm2_L2}. First, a simplex $\Delta \in \Xi$ containing $q$ is found,
for which there is a point in $\Delta$ strictly above $q$. This takes $O(\log r)$ time. 
Then, the pyramids in the conflict list of $\Delta$, that lie  
below $q$, are reported. This takes $O(Q(n) + n / r)$ time, where $Q(n)$ is the query 
time of the range counting data structure.
Let $C_\Delta$ be the set of different colors of the reported pyramids.
The final step is to count the frequencies of the colors in $C_\Delta$ among the pyramids 
below the query point, as well as counting the frequency of the color stored in $\Delta$.
This step takes $O(n^{1 + \delta} / r)$ time with a linear-sized range tree, and 
$O((n / r) \log n)$ time with a range tree of $O(n \log n)$ size. Thus:

\begin{proposition}
  \label{prop:range_mode}
  Let $P$ be a set of $n$ points in $\R^2$.
  Let $r \in [1, n]$ be a parameter and $\delta > 0$
  an arbitrarily small constant.  In $O(n \log n + nr^3)$ time, we can build a data
  structure of $O(n + r^3)$ size, that answers square range mode
  queries in $O(n^{1 + \delta} / r)$ time. Alternatively, with $O(n \log n + r^3)$
  space, the query time can be decreased to $O((n / r) \log n)$.
\end{proposition}

\subsection{Efficiently coloring cuttings}
\label{sec:rm2_prep}

In the analyses of Sections~\ref{sec:rm2_L2} and~\ref{sec:rm2_Linfty}, the
dominating term for the preprocessing time is that of computing the
colors stored in the simplices. This was done by simply scanning
through the set of planes or pyramids to find the mode color among
those below a given simplex, and repeating this procedure for all
$O(r^3)$ simplices.  The resulting procedure takes $O(nr^3)$ time,
which, for our choice of $r$ yields a quadratic time
algorithm. However, aside from this procedure, the preprocessing time
is only $O(nr^2)$. We now show how to lower the time
taken to compute the colors to
$O(n^{1 + \delta} + n^{2 / 3} \cdot r^3 + nr^2)$ for the case of
planes, and $O(n \log n + n^\delta \cdot r^3 + nr^2)$ (using $O(n)$
space) or $O(n \log n + r^3 \log n + nr^2)$ (using $O(n \log n)$
space) for the case of pyramids.  These complexities are better than
$O(nr^3)$ for certain values of $r$, which include the values that we
set $r$ to later on.

Let $S$ be a set of $n$ surfaces in $\R^3$, that are either all planes, 
or all upside-down pyramids.
Let $\Xi$ be a $\frac{1}{r}$-cutting of $S$.
The main idea behind the faster preprocessing algorithm is that we can look at the cutting 
$\Xi$ as a graph $G = (V, E)$, where $V$ is the set of vertices of $\Xi$ and $E$ is the set 
of bounded edges of $\Xi$.
The algorithm traverses this graph using depth-first search, and keeps track of the 
frequencies of all colors among those surfaces of $S$ that lie below the current 
vertex.
This information can then be converted into the color that we want to store in an incident 
simplex, using the following observation.

\begin{observation}
  \label{obs:surface_below_simplex}
    Let $\Delta \in \Xi$ be a simplex and let $v$ be a vertex of $\Delta$.
    A surface lies below $\Delta$ if and only if it lies below $v$, and it does not intersect the interior of $\Delta$.
\end{observation}

\subparagraph*{The data structure.}
For the algorithm, we use the following data structure, built on those surfaces of $S$ that lie below the current vertex.
For each color $c$, we store its frequency $f_c$ among the surfaces in the data structure.
We also keep a set of $n + 1$ linked lists $L_0, \dots, L_n$, such that list $L_f$ stores those colors whose frequency is $f$.
As we want to quickly find the location of a color inside a linked list, we store pointers for each color, pointing to the element in a linked list containing the color.
That is, for color $c$, we store a pointer to color $c$ in list $L_{f_c}$.
See Figure~\ref{fig:faster_preprocessing_data_structure_appendix} for an illustration of this part of the data structure.
\begin{figure}
\centering
\includegraphics{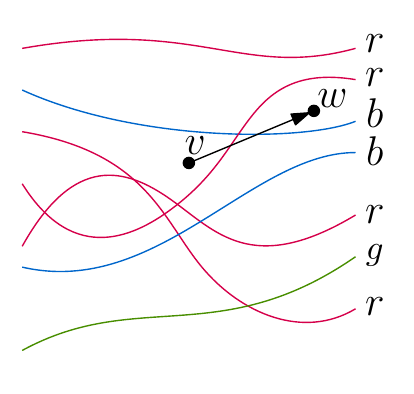}
\quad
\includegraphics{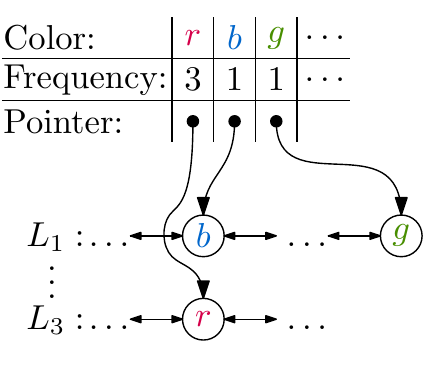}
\quad
\includegraphics{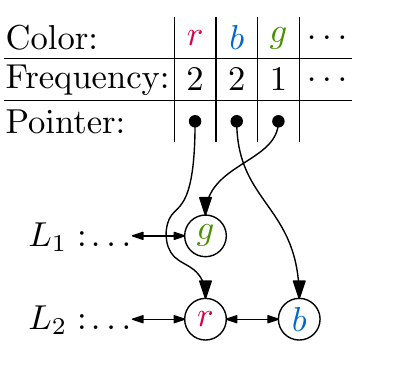}
\caption{
    The data structure used during the alternative preprocessing algorithm.
    (left) Vertex $v$ has three red surfaces, one blue surface and one green surface
    below it. The algorithm moves from $v$ to $w$, crossing a red and blue surface.
    (middle) The status of the data structure at vertex $v$.
    (right) The status of the data structure at vertex $w$.
}
\label{fig:faster_preprocessing_data_structure_appendix}
\end{figure}
In addition, we add pointers from a surface in $S$ to its color.
To keep track of the current highest frequency, we also keep the highest index $f^*$ for which $L_{f^*}$ is non-empty.
This index will be used to quickly find a mode color, as every color stored in $L_{f^*}$ will have the highest frequency.
Finally, we need a data structure that can report the surfaces in the conflict list
of a simplex $\Delta \in \Xi$, which lie below a query point $q \in \Delta$.
In the case where $S$ consists of planes, this data structure is that of Lemma~\ref{lem:cuttings_and_aux_data_structures}.
When $S$ consists of pyramids, this data structure is that of Lemma~\ref{lem:reporting_conflict_list_pyramids}.
Let $P(n)$ be the preprocessing time of this data structure, and let $Q(n, r)$ be its query time (both depending on the kind of surfaces that are in $S$).

During traversal of $G$, we insert and delete surfaces from the data structure.
This leads to us incrementing and decrementing the frequencies of the different colors.
To increment the frequency of color $c$, we first find its current frequency, which is $f_c$.
We can then find color $c$ inside list $L_{f_c}$ in $O(1)$ time, by following the pointer that we stored.
Removing the color from the linked list can be done in $O(1)$ time, now that we have the location of the color.
Then we insert color $c$ into list $L_{f_c + 1}$ and adjust the stored pointer to point to the new location of color $c$.
This can also be done in $O(1)$ time.
Finally, it might be the case that $c$ was a mode color, which would result in the frequency $f^*$ being incremented.
To account for this, we set $f^* \gets \max\{f^*, f_c + 1\}$.
The whole incrementing procedure takes $O(1)$ time in total.
Decrementing the frequency of a color works similarly, with the only real difference being how $f^*$ is updated.
This is because $f^*$ must only be updated when $f^* = f_c$ and when $L_{f_c}$ is now empty.
If this is the case, we simply set $f^* \gets f_c - 1$.
As we can check whether a list is empty in $O(1)$ time, decrementing a frequency takes $O(1)$ time as well.
We can therefore insert and delete surfaces from the data structure in $O(1)$ time.

The initial data structure can be built as follows.
Initially, every color has frequency $0$.
We therefore set every entry $f_c$ to $0$, add every color to list $L_0$, and set $f^* \gets 0$.
Say we now start the traversal at vertex $v$.
Using a linear scan over $S$, we can find those surfaces that lie below $v$.
Insert each of these surfaces into the data structure with the above procedure.
This takes $O(n)$ time in total.
The construction time of the data structure for reporting the right subsets of the conflict
list of a simplex is $P(n)$, making the total initialization time $O(P(n) + n)$.

\subparagraph*{The algorithm.}
We now present our algorithm. The algorithm is a depth-first search traversal through $G$,
where we keep the data structure updated at every vertex.
To this end, we update the data structure when traversing an edge $e = (v, w)$
by removing those surfaces that lie below $v$ and strictly above above $w$, and by inserting 
those surfaces that lie strictly above $v$ and below $w$.
See Figure~\ref{fig:faster_preprocessing_data_structure_appendix} for an illustration of how the data structure changes when traversing an edge.
Note that all removed surfaces intersect the interior of some simplex $\Delta_1$, and
that all inserted surfaces intersect the interior of some simplex $\Delta_2$, which both
have $e$ as an edge. Therefore, we can report these surfaces in $O(Q(n, r))$ time.
Note that we might report too many surfaces, but that the surfaces are all contained in
the conflict lists of $\Delta_1$ and $\Delta_2$. For each of the reported surfaces, we can check in $O(1)$ time whether it has to be removed, inserted or ignored, after which we can update the data structure in $O(1)$ time.
In total, we thus use $O(Q(n, r) + n / r)$ time to update the data structure.

When a vertex $v$ is encountered during the traversal, we assign colors to all its
incident simplices, that do not have colors stored yet.
If a simplex $\Delta$ does not have a color stored yet, we compute its color as follows.
The current data structure contains all surfaces below $v$.
By Observation~\ref{obs:surface_below_simplex}, a surface below $v$ lies below $\Delta$ if and only if it does not intersect the interior of $\Delta$.
Therefore, we report all surfaces that \textit{do} intersect the interior of $\Delta$ (the
conflict list of $\Delta$), and which lie below $v \in \Delta$, and remove these
from the data structure.
Finding and removing these surfaces takes $O(Q(n, r) + n / r)$ time.
The result is a data structure containing precisely the surfaces below $\Delta$.
Therefore, the frequency $f^*$ is the frequency of the color that we want to store in $\Delta$.
Any color stored in list $L_{f^*}$ has this frequency, so we choose the head node of the list as the color to store in $\Delta$.
For any vertex, we can thus compute the color to store in an incident simplex in $O(Q(n, r) + n / r)$ time in total.
Afterwards, we need to undo the changes made to the data structure by reinserting the surfaces that we deleted, taking $O(Q(n, r) + n / r)$ additional time.

Combining all these components, we have a graph traversal algorithm 
for computing the colors stored in the simplices.
The algorithm takes $O(Q(n, r) + n / r)$ time to traverse an edge.
Traversing all edges therefore takes $O(Q(n, r) r^3 + nr^2)$ time.
The total time spent computing the colors of incident vertices is $O(Q(n, r) r^3 + nr^2)$, as it only has to be done once per simplex.
The total time spent in the vertices is therefore $O(Q(n, r) r^3 + nr^2)$.
The result is an $O(P(n) + Q(n, r) r^3 + nr^2)$ time algorithm that computes
the mode color among the surfaces of $S$ below a given simplex in $\Xi$, for
every simplex in $\Xi$. 

In the case where $S$ consists of planes, we have that
$P(n) = O(n^{1 + \delta})$ in expectation and $Q(n, r) = O(n^{1/2} \polylog n + n / r)$, by
Lemma~\ref{lem:cuttings_and_aux_data_structures}. This gives the
following results.

\begin{lemma}
    Let $H$ be a set of $n$ colored planes in $\R^3$ and let $\Xi$ be a 
    $\frac{1}{r}$-cutting of $\A(H)$, for some $r \in [1, n]$. We can compute
    the mode color among all planes in $H$ strictly below a given simplex in $\Xi$,
    for all simplices in $\Xi$, in $O(n^{1 + \delta} + n^{1/2} r^3 \polylog n + nr^2)$
    total expected time. The extra space used is $O(n)$.
\end{lemma}

\begin{theorem}
  \label{thm:range_mode_L2}
  Let $P$ be a set of $n$ colored points in $\R^2$ and $r \in [1, n]$
  a parameter. In $O(n^{1 + \delta} + n^{1/2} r^3 \polylog n + nr^2)$ expected time, we can build a
  data structure of $O(n + r^3)$ size, that reports the mode color
  among the points in a query disk in $O(n^{2 / 3} + (n / r^{1 / 2}) \polylog n)$
  time.
\end{theorem}

In the case where $S$ consists of pyramids, we have that
$P(n) = O(n \log n)$ and either that $Q(n, r) = O(n^\delta + n / r)$
(using $O(n)$ space) or $Q(n, r) = O(\log n + n / r)$ (using
$O(n \log n)$ space), by
Lemma~\ref{lem:reporting_conflict_list_pyramids}.  This gives the
following results.

\begin{lemma}
  Let $\nabla$ be a set of $n$ colored pyramids in $\R^3$ and let
  $\Xi$ be a $\frac{1}{r}$-cutting of $\A(\nabla)$, for some
  $r \in [1, n]$. Let $\delta > 0$ be an arbitrarily small constant.
  We can compute the mode color among all pyramids in
  $\nabla$ strictly below a given simplex in $\Xi$, for all simplices
  in $\Xi$, in $O(n \log n + n^{\delta} r^3 + nr^2)$ time total. 
  The extra space used is $O(n)$. Alternatively, with $O(n \log n)$ extra space, the
  running time can be reduced to $O(n \log n + r^3 \log n + nr^2)$.
\end{lemma}

\begin{theorem}
  \label{thm:range_mode_L_infty}
  Let $P$ be a set of $n$ points in $\R^2$.
  Let $r \in [1, n]$ be a parameter and $\delta > 0$ an arbitrarily small constant.
  In $O(n \log n + n^\delta r^3 + nr^2)$ time, we can build a data
  structure of $O(n + r^3)$ size, that answers square range mode
  queries in $O(n^{1 + \delta} / r)$ time. Alternatively, with $O(n \log n + r^3)$
  space, the preprocessing time can be decreased to $O(n \log n + r^3 \log n + nr^2)$, and the query time to $O((n / r) \log n)$.
\end{theorem}

\subparagraph{Answering chromatic \kNN queries.} In case of the $L_2$
distance, we use the data structure of
Theorem~\ref{thm:range_mode_L2}, setting $r=n^{1/3}$ together with
Theorem~\ref{thm:2d_finding_query_range_L2} to answer chromatic \kNN
queries. In case of the $L_\infty$ distance we use
Theorem~\ref{thm:range_mode_L_infty} with $r=n^{1/3}$ (for $O(n)$ space)
or $r = (n \log n)^{1 / 3}$ (for $O(n \log n)$ space) together with the data
structure from
Theorem~\ref{thm:2d_finding_query_range_L_infty_linear}. We thus
obtain the following results:

\begin{theorem}
  Let $P$ be a set of $n$ points in $\R^2$. There is a linear space
  data structure that answers chromatic $k$-NN queries on $P$ with
  respect to the $L_2$ metric. Building the data structure takes expected
  $O(n^{5/3})$ time, and queries take $O(n^{5/6} \polylog n)$ time.
\end{theorem}

\begin{theorem}
  Let $P$ be a set of $n$ points in $\R^2$. There is a linear space
  data structure that answers chromatic $k$-NN queries on $P$ with
  respect to the $L_\infty$ metric. Building the data structure takes
  $O(n^{5/3})$ worst-case time, and queries take $O(n^{2/3+\delta})$
  time. We can decrease the query time to
  $O(n^{2/3}\log^{2/3} n)$ time using $O(n\log n)$ space and
  $O(n^{5/3}\log^{2/3} n)$ preprocessing time.
\end{theorem}

\section{Lower bounds}
\label{sec:Lower_bounds}

We now discuss to what extent our results may be improved further. In
particular, we relate the cost of chromatic \kNN queries to those of
range mode and range counting queries. This allows us to obtain lower
bounds for chromatic \kNN queries as well. We start with the following
reduction.

\begin{lemma}
  \label{lem:range_mode_reduction}
  Let $P$ be a set of $n$ points in $\R^d$. A range mode query with a
  query ball $\mathcal{D}_m$ (with respect to metric $m$) can be
  answered using a single range counting query with query range
  $\mathcal{D}_m$ and a single chromatic \kNN query.
\end{lemma}

\begin{proof}
  We use the range counting query to find the number of points in the
  range $\mathcal{D}_m$. Let this be $k$, and let $q$ be the center
  point of the ball $\mathcal{D}_m$. Hence, $\Dk_m(q)=\mathcal{D}_m$,
  and thus the answer to the chromatic \kNN query with center $q$ and
  value $k$ is the mode color of $\mathcal{D}_m$.
\end{proof}

\subparagraph{A conditional $\Omega(n^{1 / 2 - \delta})$ lower bound.}
We now use the above reduction to extend a conditional lower bound of
Chan~\etal~\cite{chan14linear_space_data_struc_range} on range mode
queries in arrays. They show that if in $T(n)$ time we can preprocess
an array of length $n$ into a range mode data structure with query
time $Q(n)$ then we can multiply two $\sqrt{n} \times \sqrt{n}$
boolean matrices in $O(T(n) + n + n \cdot Q(n))$ time. Boolean matrix
multiplication of two square matrices has been researched extensively
(see for example~\cite{DBLP:conf/icalp/Yu15, DBLP:conf/soda/Chan15,
  DBLP:conf/focs/FischerM71, DBLP:conf/focs/BansalW09}). The best
combinatorial bounds merely shave off logarithmic factors from the
trivial time bound, which is $O(n^{3 / 2})$ for two
$\sqrt{n} \times \sqrt{n}$ matrices.
Hence, this shows strong evidence
that the query time for range mode queries must either have a
worst-case preprocessing time that is $\omega(n^{3 / 2 - \delta})$, or
a worst-case query time that is $\omega(n^{1 / 2 - \delta})$, assuming
only combinatorial techniques may be used. As we argue next, the same
holds for the query time of chromatic \kNN queries of $n$ points in
$\R^1$ in a $L_m$-metric.

Consider an array $A$. We map each entry $A[i]$ to the point $p_i = i$
in $\R^1$, and assign this point color $A[i]$. A query now asks for the
mode element inside a subarray $A[\ell \mathrel{:} r]$. We answer this
query using the approach in Lemma~\ref{lem:range_mode_reduction}. That
is, we take query range $\mathcal{D}=[p_\ell,p_r]$, which is a disk
in $\R^1$, and can be computed in constant time from the query pair
$\ell,r$. Furthermore, since we know $\ell$ and $r$ we can answer the
range counting query on $\mathcal{D}$ in constant time by simply
reporting $k = r - \ell + 1$. By Lemma~\ref{lem:range_mode_reduction}
it thus follows that if answering the \kNN query takes $O(Q(n))$ time
and uses $S(n)$ space, so does answering the range mode query. This
implies the following result (following
Chan~\etal~\cite{chan14linear_space_data_struc_range}).

\begin{theorem}
  \label{thm:sqrt_lowerbound}
  Assume that in $T(n)$ time, we can build a data structure that
  answers chromatic $k$-nearest neighbors queries on a set of $n$
  points in $\R^1$ under any $L_m$ metric in $Q(n)$ time. We can then
  perform boolean matrix multiplication of two
  $\sqrt{n} \times \sqrt{n}$ matrices in $O(T(n) + n + n \cdot Q(n))$
  time.
\end{theorem}

Theorem~\ref{thm:sqrt_lowerbound} thus shows strong evidence that
using near-linear space and preprocessing time chromatic \kNN queries
require $\Omega(n^{1 / 2 - \delta})$ time.

\subparagraph{Relations to range counting queries.} Next, we relate
the cost of range finding queries, i.e. the problem solved in our
first step, to range counting queries. Given a data structure for
range finding queries we can answer range counting queries using only
logarithmic overhead:

\begin{lemma}
  \label{lem:lower_bound_range_finding}
  Let $P$ be a set of $n$ points in $\R^d$. A range counting query on
  $P$ with a disk $\mathcal{D}_m$ under metric $m$ can be performed
  using $O(\log n)$ range finding queries.
\end{lemma}
\begin{proof}
  Let $q$ be the center of the query disk. We binary search over the
  integers $0,\dots,n$, using a range finding query to find a disk $\Dk_m(q)$ for each considered
  integer $k$. If the reported disk is smaller than $\mathcal{D}_m$,
  the number of points inside $\mathcal{D}_m$ is at least $k$.
  Otherwise, the number of points is smaller than $k$. It follows that
  with $O(\log n)$ range finding queries, we can count the number of
  points in $\mathcal{D}_m$.
\end{proof}

It thus follows that range finding is roughly as difficult as range
counting. In particular, a $Q(n)$ time lower bound for range counting
queries using $S(n)$ space implies an $\tilde{\Omega}(Q(n))$ time
lower bound for range finding queries with $S(n)$ space. For example,
in the semigroup model there is an $\tilde{\Omega}(n/S(n)^{1/d})$ time
lower bound for halfspace range
counting~\cite{arya2012halfspace}. Since every halfspace is a disk
$\mathcal{D}_2$ (of radius $\infty$), this lower bound also holds for
range counting with disks in the $L_2$ metric, and thus also for range finding.

The range mode queries from step 2 are also related to a form of range
counting. A ``type-2'' range counting query with query range $Q$ asks
for all the distinct colors appearing in $Q$ together with their
frequencies, i.e. for each reported color $c$ we must also report the number of points in
$P\cap Q$ that have color
$c$~\cite{chan20furth_resul_color_range_searc}. Clearly, answering
``type-2'' queries is more difficult than range counting (just
assign all points the same color), so the above lower bounds also hold
for ``type-2'' queries. Such ``type-2'' queries however also
allow us to solve the range mode problem. When the number of colors is
small (e.g. two), and we already know the number of points $k$ in the
query range $\Dk(q)$ it seems that answering range mode queries is not
much easier than answering (``type-2'') range counting queries. We
therefore conjecture that answering range mode queries is roughly as
difficult as answering range counting queries.

\begin{conjecture}
  \label{con:range_mode}
  If answering a range counting query with a query range $\mathcal{D}$
  using $S(n)$ space requires $Q(n)$ time then answering a range mode
  query with query range $\mathcal{D}$ using $S(n)$ space requires
  $\tilde{\Omega}(Q(n))$ time.
\end{conjecture}

Note that this conjecture together with
Lemma~\ref{lem:range_mode_reduction} would imply that answering a \kNN
query is at least as hard as answering a range counting
query. Furthermore, since we can answer a range counting query using
$O(\log n)$ range finding queries
(Lemma~\ref{lem:lower_bound_range_finding}) that would then mean our
two-step approach has negligible overhead with respect to an optimal
solution to chromatic \kNN queries.

\section{The approximate problems}
\label{sec:approximation}

In this section we consider $\eps$-approximate chromatic \kNN
queries. Our goal is to report a color $c$ that occurs at least
$(1-\eps)f^*$ times, where $f^*$ is the frequency of the mode color
$c^*$ of $\kNN(q)$. We again use the two-step approach of finding the
range $\Dk(q)$ (step (1)) and computing the mode of the range (step
(2)). We use exact range finding data structures, and focus our
attention on approximating step (2) for two reasons: first, the
running times in our exact solutions are dominated by step (2), and
second, it is unclear how to use approximate solutions to \kNN queries
(that is, approximate ranges) and still obtain guarantees on the
approximation factor of our $\eps$-approximate chromatic \kNN queries.

For points in $\R^1$ we can again compute $\Dk(q)$ exactly using the
Theorem~\ref{thm:1d_finding_query_range} data structure, and then
directly use the $O(n/\eps)$ space data structure by Bose
\etal~\cite{DBLP:conf/stacs/BoseKMT05} to answer the remaining
$(1-\eps)$-approximate range mode query. We thus get:

\begin{theorem}
  \label{thm:approx_1d}
  Let $P$ be a set of $n$ points in $\R$ and $\eps \in (0, 1)$ a parameter.
  In $O(n \log_{\frac{1}{1 - \eps}} n)$ time, we can build a data
  structure of $O(n / \eps)$ size, that answers approximate chromatic
  \kNN queries on $P$ under any $L_m$ metric, with $m \geq 1$, in
  $O(\log n + \log \log_{\frac{1}{1 - \eps}} n)$ time.
\end{theorem}

We now focus on the problem for points in $\R^2$. We can again use our
data structures from Section~\ref{sec:finding_the_range_2d} to
efficiently find the range $\Dk(q)$ containing $\kNN(q)$. In
Section~\ref{sub:approximate_L_2} we now show that we can answer
$\eps$-approximate range queries with disks in roughly
$\Otilde(n^{1/2}\eps^{-3/2})$ time (ignoring polylogarithmic
factors), while still using near linear space. In
Section~\ref{sub:approx_L1_Linfty} we extend this approach to the
$L_\infty$ metric.

\subsection{Approximate chromatic \kNN queries under the $L_2$ metric}
\label{sub:approximate_L_2}

We develop a data structure storing $P$ that can efficiently answer
$\eps$-approximate range mode queries with a query disk $Q$ of radius
$r$. To this end, we again transform $P$ into a set of planes $H$ like
in
Section~\ref{sec:rm2_L2}. The
query disk now corresponds to a vertical halfline with top
endpoint $h^*$, the point dual to the plane $h : z = 2q_x x + 2q_y y - q_x^2 - q_y^2 + r^2$. So, the mode color $c^*$
of $P \cap Q$ is the most frequently occurring color among the planes
passing below $h^*$, and our aim is to report a color $c$ such
that at least $(1-\eps)f^*$ planes of that color pass below $h^*$.

The \textit{$k$-level} of the arrangement $\mathcal{A}(H)$ of planes
is the set of points that lie on a plane in $H$ and which have exactly
$k$ planes passing strictly below them. An \textit{$\eps$-approximate
  $k$-level} of $\mathcal{A}(H)$ is a piecewise linear triangulated
terrain, so every vertical line intersects it once, that lies in
between the $k$-level and $(1 + \eps)k$-level of
$\mathcal{A}(H)$. Har-Peled~\etal~\cite{Har-Peled2017} show that such
an $\eps$-approximate $k$-level exists, has complexity
$O\left( \frac{n}{\eps^3 k} \right)$, and can be computed in
$O(n(\frac{1}{\eps^{3}}+\frac{1}{k\eps^6}\log^3\frac{n}{k} +
\log\frac{n}{k\eps}))$ expected time.\footnote{Har-Peled~\etal describe
  a Monte Carlo algorithm to compute an $\eps$-approximate $k$-level
  in $O(n+\frac{n}{k\eps^6}\log^3(n/k))$ time. To turn this algorithm
  into a Las Vegas algorithm we have to compute the conflict lists of
  the prisms defined by the approximate level. This requires an
  additional $O(\frac{n}{\eps^3}+n\log\frac{n}{\eps k})$ time.}

The main idea to answer approximate range mode queries efficiently is
to compute, for each color $c$, a series of $g(\eps)$-approximate
$k_i$-levels (for some function $g$) considering only the planes of
color $c$. For each choice of $i$, we then consider the lower envelope
$\LL_i$ of all those $k_i$-levels among the various
colors. See Figure~\ref{fig:approximation} for an illustration. Now observe
that if $h^*$ lies in between $\LL_i$ and $\LL_{i+1}$, the
frequency $f^*$ of a mode color $c^*$ may only be a $g(\eps)$ fraction
larger than $k_{i+1}$, while the frequency of the color defining
$\LL_i$ directly below $h^*$ is at least $k_i$. So if $g(\eps)$
and $k_i/k_{i+1}$ are sufficiently small this is a
$(1-\eps)$-approximation. One additional complication is that even
though our $g(\eps)$-approximate $k_i$-levels have fairly small
complexity, their lower envelopes do not. So, we need to design a data
structure that can test if $h^*$ lies above or below $\LL_i$
without explicitly storing $\LL_i$. We show that with near-linear
space we can answer such queries in $O_\eps(n^{1/2})$ time.
\begin{figure}[tb]
  \centering
  \includegraphics[page=1]{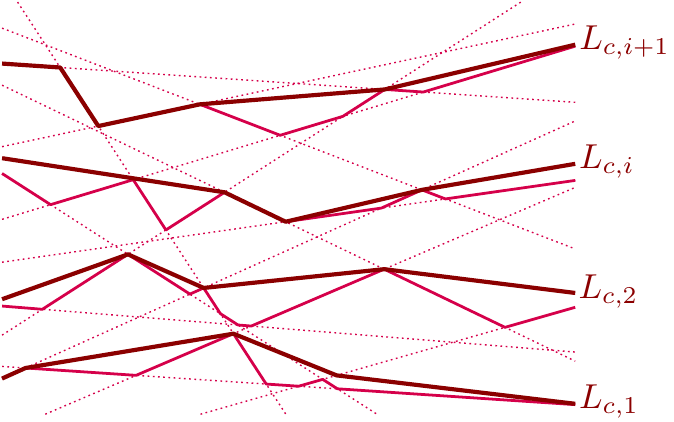}
  \quad
  \includegraphics[page=2,clip,trim=0 0 0.5cm 0]{approximate_levels}
  \caption{(left) An illustration of the idea in $\R^2$, the planes
    (here lines) of a single color, their
    $k_i=\left(\frac{1}{1-\alpha}\right)^i$-levels (bright red), and
    the $g(\eps)$-approximate $k_i$-levels $L_{c,0},L_{c,1},\dots$ (in
    dark red). (right) For each $i$, the $\LL_i$ forms the lower
    envelope of the $L_{c,i}$ surfaces over all colors $c$. We search
    for the largest $i$ for which $q$ lies above $\LL_{i}$ (dashed).  }
  \label{fig:approximation}
\end{figure}

Set $\alpha = 1 - \sqrt{1 - \eps}$, and let $H_c$ be the set of planes
in $H$ with color $c$. For each set $H_c$, we consider the
$\frac{\alpha}{1 - \alpha}$-approximate
$\left( \frac{1}{1 - \alpha} \right)^i$-levels $L_{c, i}$ of
$\A(H_c)$, for $i = 0, \dots, \log_{\frac{1}{1 - \alpha}} |H_c|$, and
we define $\LL_i$ to be the lower envelope (the $0$-level) of the
(arrangement of the) surfaces $\LL_{c,i}$ over all colors $c$. Hence
we have $k_i=\left( \frac{1}{1 - \alpha} \right)^i$, and
$g(\eps)=\frac{1-\sqrt{1-\eps}}{\sqrt{1-\eps}}$. For any point $p$,
let $\LL_i(p)$ be the point of intersection between $\LL_i$ and the
vertical line through $p$. The following lemma now states that we can
use these approximate levels to find an approximate mode color.

\begin{lemma}
  \label{lem:monotoincially_increasing}
  For any $\alpha > 0$ and any positive integer $i$, the surface
  $\LL_{i+1}$ lies above $\LL_i$.
\end{lemma}
\begin{proof}
    For any color $c$, the surface $L_{c,i}$ lies below the
    \[
        \left(1+\frac{\alpha}{1-\alpha}\right)\left( \frac{1}{1-\alpha} \right)^i
      = \left(\frac{1-\alpha}{1-\alpha} +\frac{\alpha}{1-\alpha}\right)\left(
        \frac{1}{1-\alpha} \right)^i
      = \frac{1}{1-\alpha}\left(
        \frac{1}{1-\alpha} \right)^i = \left( \frac{1}{1-\alpha} \right)^{i+1}
    \]
    level of $\A(H_c)$. Since $L_{c,i+1}$ is an approximate
    $\left( \frac{1}{1-\alpha} \right)^{i+1}$-level, it thus follows
    that $L_{c,i}$ lies below $L_{c,i+1}$. Since this holds for any
    color, it also holds for $\LL_i$ and $\LL_{i+1}$: for any point $p$
    on $\LL_{i+1}$, let $L_{c,i+1}$ be the surface realizing $\LL_{i+1}$
    at $p$. By the above argument it then follows that $L_{c,i}(p)$ lies
    below $\LL_{i+1}(p)$, and hence $\LL_i$ also lies below
    $\LL_{i+1}(p)$.
\end{proof}

\begin{lemma}
  \label{lem:approximate_mode}
  Let $q \in \R^3$ be a point, let $f^*$ be the exact frequency of
  the mode color among the planes below $q$, and let $i$ be the
  largest integer for which $q$ lies above $\LL_i$. Any color $c$ for
  which $q$ lies above $L_{c,i}$ occurs at least $(1 - \eps)f^*$
  times among the planes below $q$.
\end{lemma}

\begin{proof}
  Because $i$ is the largest integer for which $q$ lies above $\LL_i$
  it follows that, for any color $c$, the point $q$ lies strictly
  below $L_{c, i + 1}$ (otherwise it would also have been above
  $\LL_{i+1}$). This holds in particular for the mode color
  $c^*$. Since $q$ lies below
  $L_{c^*,i+1}$, it lies strictly below the
  $\left( \frac{1}{1 - \alpha} \right)^{i + 2}$-level of
  $\A(H_{c^*})$. We obtain
  $\displaystyle
  f^* < \left( \frac{1}{1 - \alpha} \right)^{i + 2} + 1 = \left( \frac{1}{1
    - \eps} \right)^{i / 2 + 1} + 1
  $ and thus
  \begin{equation}
    \label{eq:upperbound_fstar}
    (1-\eps)f^* < \left( \frac{1}{1
    - \eps} \right)^{i / 2} + (1-\eps).
 \end{equation}

  For any color $c$ for which $q$ lies above $L_{c,i}$ we have that
  $q$ lies above the $\left( \frac{1}{1 - \alpha} \right)^i$-level of
  $\A(H_c)$. We thus get the following lower bound for the number of
  planes of color $c$ below $q$:
  \begin{align}
    \label{eq:lowerbound_fc}
    \left( \frac{1}{1 - \alpha} \right)^i + 1 = \left( \frac{1}{1 -
        \eps} \right)^{i / 2} + 1 \leq f_c.
  \end{align}
  Combining equations~\ref{eq:upperbound_fstar}
  and~\ref{eq:lowerbound_fc} we obtain that for a color $c$ for which
  $q$ lies above $L_{c,i}$ we have
  \[
    (1-\eps)f^* < \left( \frac{1}{1
        - \eps} \right)^{i / 2} + (1-\eps) < \left( \frac{1}{1 -
        \eps} \right)^{i / 2} + 1 \leq f_c.
  \]
  Hence, $f_c \geq (1-\eps)f^*$. Note that there is at least one color
  $c$ for which $q$ lies above $L_{c,i}$, and thus for which
  $f_c \geq (1-\eps)f^*$, namely the color defining $\LL_i$ at $\LL_i(h^*)$.
\end{proof}

By Lemma~\ref{lem:approximate_mode} we can answer an
$\eps$-approximate range mode color query by finding the largest
integer $i$ among $0, \dots, \log_{\frac{1}{1 - \alpha}} n$ for which
$h^*$ lies above $\LL_i$. By
Lemma~\ref{lem:monotoincially_increasing}, all
$\LL_0(h^*),\LL_1(h^*),\dots$ are monotonically
increasing. Hence, we can find $i$ using a binary search. Note that if 
$h^*$ lies below $\LL_0$, every color will have a frequency of at most one.
We can answer an \emph{exact} chromatic \kNN query in this case with
range reporting, reporting the color of an arbitrary point in the range. 
This step does not affect the asymptotic complexities.

Now assume that $h^*$ lies above $\LL_0$. To find the largest integer $i$
for which $h^*$ lies above $\LL_i$, we
need to be able to compute a points $\LL_j(h^*)$ (for some
given $j$). We now argue that we can do this efficiently while using
only near linear space.

\subparagraph{Computing $\LL_j(h^*)$.} One solution would be to
explicitly compute the lower envelope of (the triangles making up) the
$L_{i,c}$ surfaces, over all colors $c$. However, this envelope may
have quadratic
size~\cite{DBLP:journals/dcg/Edelsbrunner89,DBLP:journals/dcg/PachS89}. We
will therefore use a ray shooting based approach
instead~\cite{agarwal94range_searc_semial_sets, DBLP:books/sp/Berg93}.

\begin{lemma}
  \label{lem:ray_shooting}
  Let $\T$ be a set of $m$, possibly intersecting, possibly unbounded,
  triangles in $\R^3$. In $O(m\log^4 m)$ time, we can build a data
  structure of size $O(m\log^2 m \log \log m)$ that can report the lowest triangle
  in $\T$ along a vertical line $\ell_q$ in $O(m^{1/2}\log^2 m)$
  expected time.
\end{lemma}
\begin{proof}
    The main idea is to project the triangles onto the horizontal plane
    $z=0$, and store this set $\bar{\T}$ of projected triangles in a
    data structure for stabbing queries. In particular, a data structure
    that can return the triangles stabbed by a query point $q$ as few
    ($O(m^{1/2}\log^2 m)$) \emph{canonical subsets} $\bar{\T'}$, i.e. fixed subsets
    of triangles stored in the data structure. On the domain where the
    $q$ stabs all triangles in such a canonical subset $\bar{\T'}$ we
    can treat their original triangles $\T'$ as planes. Hence, using
    linear space we can store the lower envelope of these planes so that
    we can efficiently (i.e. in $O(\log |\bar{\T'}|)$ time) report the lowest
    plane (and thus triangle) of $\bar{\T'}$ intersected by
    $\ell_q$. We then report the lowest triangle intersected over all
    canonical subsets. Next, we describe the implementation of our data
    structure in more detail.

    Every projected triangle $\bar{\Delta}$ is the intersection of (at most)
    three halfplanes $h_1$,$h_2$, and $h_3$. We dualize the lines
    bounding these halfplanes into three points $p_1(\bar{\Delta})$,
    $p_2(\bar{\Delta})$, and $p_3(\bar{\Delta})$, respectively, and the query
    point $q$ into a line $\rho$. The query point $q$ lies in $h_1$ if
    and only if $\rho$ passes $p_1(\bar{\Delta})$ on the side corresponding
    to $h_1$. It thus follows that $q$ lies in the triangle $\bar{\Delta}$ if
    and only if $\rho$ passes $p_1(\bar{\Delta})$, $p_2(\bar{\Delta})$, and
    $p_3(\bar{\Delta})$ on the appropriate sides. We build a three-level
    partition tree on these points~\cite{chan2012optimal}. In
    particular, the first level is a partition tree on the
    $p_1(\bar{\Delta})$ points over all over all triangles $\Delta \in T$. Each
    internal node $\nu$ corresponds to a subset of the points, and thus
    a subset $\T_\nu$ of triangles, (namely the points stored in the
    leaves of the subtree) and stores a partition tree on the
    $p_2(\bar{\Delta})$ points of a triangle $\Delta \in \T_\nu$. In turn, the
    internal nodes of this tree store one more partition tree on the
    $p_3$ points. Each node $\mu$ in a third level partition tree
    corresponds to some region $R_\mu$ of the (primal) plane (at $z=0$)
    that is contained in all of the (projections of the) triangles of
    its associated set $T_\mu$. Hence, $\mu$ stores the lower envelope
    of the supporting planes of the triangles in $\T_\mu$ preprocessed
    for $O(\log |\T_\mu|)$ time vertical ray shooting queries. Since
    this lower envelope has linear complexity, such a data structure
    (essentially a point location data structure) uses only linear
    space~\cite{snoeyink04pointlocation}.
    We apply Corollary 7.2(i) of Chan~\cite{chan2012optimal} for the trees
    in the first two levels, and then Corollary 7.2(ii) for the third level
    trees. It then follows that the total space used is $O(m \log^2 m \log \log m)$,
    and that the data structure can be built in $O(m \log^4 m)$ time.
    For a query $q$, we can select (in expectation) $O(m^{1 / 2} \log^2 m)$
    nodes $\mu$ of the third level trees whose canonical subsets $\T_\mu$
    make up exactly the triangles stabbed by $q$~\cite{chan2012optimal}.
    For each such a node $\mu$ we thus have $q \in R_\mu$, and hence we can find 
    the lowest triangle vertically above $q$ by querying the lower envelope 
    of the planes in $O(\log |\T_\mu|)$ time. As the canonical subsets can 
    be selected in $O(m^{1 / 2} \log^2 m)$ expected time, and have the property 
    that $\sum_\mu \log \T_\mu = O(m^{1 / 2} \log^2 m)$, it follows that
    the total expected query time is $O(m^{1/2}\log^2 m)$.
\end{proof}

\begin{remark}
    Chan remarks that the above bound can likely be made to hold
    w.h.p. (i.e. probability $1-\frac{1}{m^{c_0}}$, for an arbitrarily large
    constant $c_0$)~\cite{chan2012optimal}. If we require worst case query time,
    we can also use the partition tree by
    \Matousek~\cite{matousek93range}. This increases the preprocessing time to 
    $O(m^{1+\delta})$, the space to $O(n \polylog n)$ and the query time to
    $O(m^{1/2}\polylog m)$, for some arbitrarily small $\delta > 0$.
\end{remark}

\subparagraph{Analysis.} We now analyze the space usage and the
preprocessing and query time.

\begin{lemma}
  \label{lem:complexity_levels}
  The total complexity of the approximate levels is $O(n / \alpha^4)$,
  and they can be computed in
  $O\left(n \left(\frac{\log_{\frac{1}{1 - \alpha}} n}{\alpha^3} +
      \frac{\log^3 n}{\alpha^7} + \log \left( \frac{n}{\alpha}
      \right)\log_{\frac{1}{1 - \alpha}} n \right) \right)$ expected
  time.
\end{lemma}
\begin{proof}
    Let $A$ be a sufficiently large constant. The total complexity of
    the approximate levels is then
    \begin{align*}
      \sum_{c} \sum_{i = 0}^{\log_{\frac{1}{1 - \alpha}} |H_c|}
      \frac{A\cdot|H_c|}{\left(\frac{\alpha}{1-\alpha}\right)^3 \cdot \left(
      \frac{1}{1 - \alpha} \right)^i}
      &=
        A\sum_{c} \sum_{i = 0}^{\log_{\frac{1}{1 - \alpha}} |H_c|}
        \frac{|H_c|}{\frac{\alpha^3}{(1-\alpha)^3} \cdot \left(
        \frac{1}{1 - \alpha} \right)^i}\\
      &= A \sum_{c} \frac{|H_c|}{\alpha^3}
        \sum_{i = 0}^{\log_{\frac{1}{1 - \alpha}} |H_c|}
        \frac{1}{\frac{1}{(1-\alpha)^3} \cdot
        \frac{1}{(1 - \alpha)^i}}\\
      &\leq A \frac{n}{\alpha^3}
        \sum_{i = 0}^{\log_{\frac{1}{1 - \alpha}} n}
        \frac{1}{\frac{1}{(1-\alpha)^{i+3}}}\\
      &= A \frac{n}{\alpha^3}
        \sum_{i = 0}^{\log_{\frac{1}{1 - \alpha}} n}
        (1-\alpha)^{i+3}.
    \end{align*}
    Using that $\sum_{i=0}^\infty (1-\alpha)^i$ is a geometric series
    that converges to $1/(1-(1-\alpha))=1/\alpha$ it then follows that
    the total complexity of the approximate levels is $O(n/\alpha^4)$.
  
    Recall that an $\eps$-approximate $k$-level can be computed in
    $O(n(\frac{1}{\eps^3}+\frac{1}{k\eps^6}\log^3\frac{n}{k} +
    \log\frac{n}{k\eps}))
    = O(n(\frac{1}{\eps^3}+\frac{1}{k\eps^6}\log^3 n +
    \log\frac{n}{\eps}))
    $
    expected time~\cite{Har-Peled2017}. The total
    time to construct our approximate levels is then
    \begin{align*}
      &\sum_c \sum_{i = 0}^{\log_{\frac{1}{1 - \alpha}} |H_c|} O\left(|H_c|
        \left( \frac{1}{\left(\frac{\alpha}{1-\alpha} \right)^3}
             + \frac{1}{
                 \left( \frac{1}{1-\alpha} \right)^i \left( \frac{\alpha}{1-\alpha}  \right)^6
                       }
        \log^3 |H_c|
        + \log \left( \frac{|H_c|}{\frac{\alpha}{1-\alpha}} \right)
          \right)  \right) \\
      &= O\left(\sum_c |H_c|\sum_{i = 0}^{\log_{\frac{1}{1 - \alpha}} n}
        \left( \frac{1}{\frac{\alpha^3}{(1-\alpha)^3}}
      + \frac{1}{\frac{\alpha^6}{(1-\alpha)^{i+6}}}
      \log^3 n
      + \log \left( \frac{n(1-\alpha)}{\alpha}  \right)
        \right)  \right) \\
      &= O\left(n\sum_{i = 0}^{\log_{\frac{1}{1 - \alpha}} n}
        \left( \frac{(1-\alpha)^3}{\alpha^3}
      + \frac{(1-\alpha)^{i+6}}{\alpha^6}
      \log^3 n
      + \log \left( \frac{n}{\alpha}  \right)
        \right)  \right) \\
      &= O\left(n
        \left(\sum_{i = 0}^{\log_{\frac{1}{1 - \alpha}} n} \frac{(1-\alpha)^3}{\alpha^3}
      + \sum_{i = 0}^{\log_{\frac{1}{1 - \alpha}} n}\frac{(1-\alpha)^{i+6}}{\alpha^6}
      \log^3 n
      + \sum_{i = 0}^{\log_{\frac{1}{1 - \alpha}} n} \log \left( \frac{n}{\alpha}  \right)
        \right)  \right) \\
      &= O\left(n
        \left(\frac{\log_{\frac{1}{1 - \alpha}} n}{\alpha^3}
      + \frac{\log^3 n}{\alpha^7}
      + \log \left( \frac{n}{\alpha}  \right)\log_{\frac{1}{1 - \alpha}} n
        \right)  \right). \qedhere
    \end{align*}
\end{proof}

It now follows from Lemma~\ref{lem:complexity_levels} that the
ray-shooting data structures from Lemma~\ref{lem:ray_shooting} have
total size
$O(\frac{n}{\alpha^4}\log^2 \left( \frac{n}{\alpha^4} \right) \log \log \left( \frac{n}{\alpha^4} \right))$, and
that they can be built in expected
$O(\frac{n}{\alpha^4}\log^4 \left( \frac{n}{\alpha^4} \right))$
time. The ray shooting data structure on $\LL_i$ is built on at most
\[ O\left( \sum_c \frac{|H_c|}{\left( \frac{\alpha}{1-\alpha} \right)^3k_i } \right)
   = O\left( \frac{n}{\left( \frac{\alpha}{1-\alpha} \right)^3} \right) =
   O(n/\alpha^3)
 \] triangles. It then follows that the expected time to answer a
 chromatic range mode query is
 $O\left(\frac{n^{1/2}}{\alpha^{3/2}}\log^2
   \frac{n}{\alpha}\log\log_{\frac{1}{1-\alpha}}n\right)$. Using
 that $\eps/2 \leq \alpha = 1 - \sqrt{1-\eps} \leq \eps$ we then also
 obtain the following result.

\begin{theorem}
  \label{thm:approx_range_mode_queries_l2}
  Let $P \subset \R^2$ be a set of $n$ points and let
  $\eps \in (0, 1)$.  In expected
  \[
  O\left(n \left(\frac{\log_{\frac{1}{1 - \eps}} n}{\eps^3} +
      \frac{\log^3 n}{\eps^7} + \log \left( \frac{n}{\eps}
      \right)\log_{\frac{1}{1 - \eps}} n + \frac{\log^4
        \frac{n}{\eps}}{\eps^4} \right) \right) 
     \quad \mathit {time,}
  \]
  we can build a
  $O\left(\frac{n}{\eps^4}\log^2\frac{n}{\eps} \log \log
    \frac{n}{\eps}\right)$ 
  size data structure    
    that answers $\eps$-approximate
  chromatic range mode queries disks in expected
  $O\left(\frac{n^{1/2}}{\eps^{3/2}}\log^2
    \frac{n}{\eps}\log\log_{\frac{1}{1-\eps}}n\right)$ time.
\end{theorem}

\subparagraph{Answering $\eps$-approximate chromatic \kNN queries.}
Since we can find the smallest disk $\Dk(q)$ containing $\kNN(q)$, and
thus the point dual to $\Dk(q)$, in $O(n^{1/2}\polylog n)$ time using
Theorem~\ref{thm:2d_finding_query_range_L2}, we can then also answer
$\eps$-approximate chromatic \kNN queries:

\begin{theorem}
  \label{thm:approx_chromatic_L2}
  Let $P \subset \R^3$ be a set of $n$ points and let
  $\eps \in (0, 1)$.  In expected
  \[
    O\left(n^{1+\delta} + n \left(\frac{\log_{\frac{1}{1 - \eps}} n}{\eps^3} +
      \frac{\log^3 n}{\eps^7} + \log \left( \frac{n}{\eps}
      \right)\log_{\frac{1}{1 - \eps}} n + \frac{\log^4
        \frac{n}{\eps}}{\eps^4} \right) \right)
     \quad \mathit {time,}
  \]
  we can build a 
  $O\left(\frac{n}{\eps^4}\log^2\frac{n}{\eps} \log \log \frac{n}{\eps}\right)$ 
  size data structure
  that answers
  $\eps$-approximate chromatic \kNN queries 
  on $P$ 
  under the $L_2$
  metric in expected
  \[
    O\left(n^{1/2}\polylog n + \frac{n^{1/2}}{\eps^{3/2}}\log^2
    \frac{n}{\eps}\log\log_{\frac{1}{1-\eps}}n\right) 
    \quad \mathit{ time.}
  \]
\end{theorem}

\subsection{Approximate chromatic \kNN queries under the $L_\infty$ metric}
\label{sub:approx_L1_Linfty}

In this section we show that the data structure from
Section~\ref{sub:approximate_L_2} can be adapted to answer
$\eps$-approximate range mode queries with axis aligned
squares. Therefore, we can also efficiently answer $\eps$-approximate
chromatic \kNN queries in the $L_\infty$ metric. The main
idea remains the same as before:
(i) we transform the points into surfaces in $\R^3$ and the query
  range $Q$ into a vertical downward half-line,
(ii) we compute a series of $g(\eps)$-approximate $k_i$-levels on
  the surfaces of each color, and
(iii) for each $i$ we store the lower envelope $\LL_i$ of those
  approximate levels over all colors.
Via the same argument as before it then follows that if the top
endpoint of the vertical half-line lies in between $\LL_i$ and
$\LL_{i+1}$, the color realizing $\LL_i$ is an answer to the
$\eps$-approximate range mode query.

Recall from
Section~\ref{sec:rm2_Linfty}
that (i) we can map each point $p \in P$ to the graph of
$L_\infty(p, \cdot)$, which we refer to as a \emph{pyramid}, and that
(ii) an axis parallel query square $Q$ of radius $r$ then corresponds
to the vertical halfline with topmost point $\hat{q}=(q_x,q_y,r)$. Let
$\nabla$ denote the set of all such pyramids. As in the case for
planes these pyramids define an arrangement $\A(\nabla)$, and an
($\eps$-approximate) $k$-level in this arrangement. Since the
$L_\infty$-Voronoi diagram on $P$ has linear
complexity~\cite{lee80two_dimen_voron_diagr_lp_metric}, so does the
$0$-level of
$\A(\nabla)$. Kaplan~\etal~\cite{DBLP:journals/dcg/KaplanMRSS20} then
show that an $\eps$-approximate $k$-level has complexity
$O(\frac{n}{\eps^5k}\log^2 n)$. Moreover, they show how to construct
an $\frac{1}{2}$-approximate $k$-level in
$O(n\lambda_s(\log n)\log^2 n)$ time, where $\lambda_s(m)$ denotes the
maximum length of a Davenport-Schinzel sequence of order $s$ on $m$
symbols. The value of $s$ depends on the type of surfaces under
consideration (refer to
Kaplan~\etal~\cite{DBLP:journals/dcg/KaplanMRSS20} for details). In
our case, we have $s=4$ (see~\cite[Section
9]{DBLP:journals/dcg/KaplanMRSS20}), and hence $\lambda_4(m)$ is near linear in $m$. Next,
we argue that we can also use the algorithm of Kaplan~\etal to compute
$\eps$-approximate levels for different values $\eps >
0$.

\begin{lemma}
  \label{lem:approximate_klevel_pyramids}
  Let $\nabla$ be a set of $n$ pyramids.  Let $1 \leq k \leq n$ and
  let $\eps \in (0, 1 / 2]$.  In
  $ O\left( \frac{n}{\eps^2}\lambda_4\left(
      \frac{\log n}{\eps^2} \right)\log^3 n \right) $ expected time, we can
  build an $\eps$-approximate $k$-level of $\A(\nabla)$ of
  complexity $O\left( \frac{n}{\eps^5 k} \log^2 n \right)$.
\end{lemma}

\begin{proof}
  This mostly follows from Theorem 8.1 of Kaplan
 ~\etal~\cite{DBLP:journals/dcg/KaplanMRSS20}, which deals with
  constructing the approximate terrain. They present their
  construction for $\eps = 1/2$. For completeness, we repeat the
  argument, showing that it also works for $\eps \in (0, 1 / 2]$.

  Let $\nabla_k \subseteq \nabla$ be a random sample of size
  $\min\left\{ \frac{cn}{\eps^2 k} \log n, n \right\}$, where $c > 0$
  is a suitable constant.  Pick $t$ uniformly at random from the range
  $\left[ \left( 1 + \frac{\eps}{3} \right) \tau, \left( 1 +
      \frac{\eps}{2} \right) \tau \right]$, where
  $\tau = \frac{c}{\eps^2} \log n$.  Kaplan~\etal prove that the
  $t$-level of $\A(\nabla_k)$ is a terrain lying between the $k$-level
  and $(1 + \eps)k$-level of $\A(\nabla)$ with high probability.  The
  expected complexity of this level is
  $ O\left(\frac{n \log^2 n}{\eps^5 k} \right)$,  and it can be
  constructed in
  $ O\left( nt \lambda_4(t) \log \left( \frac{n}{t} \right) \log n
  \right) = O\left( \frac{n \log^3 n}{\eps^2}\lambda_4 \left(
      \frac{\log n}{\eps^2} \right) \right) $ expected time using a
  randomized incremental construction approach.

  Let $\overline{T}_k$ be the $t$-level of $\A(\nabla_k)$.  To check
  whether $\overline{T}_k$ is really a terrain between the $k$-level
  and $(1 + \eps)k$-level, we need to test if for each triangle
  $\Delta$ of $\overline{T}_k$, the number of pyramids from $\nabla$
  that intersect the prism with top facet $\Delta$ (i.e. the region
  $\{(x,y,z) \mid (x,y,z') \in \Delta \land z \leq z' \}$) is at most
  $(1+\eps)k$. This can be done in $ O\left( \frac{n \log^3 n}{\eps^5}
  \right)$ time based on the conflict lists computed during the
  randomized incremental construction.

  We stop, and repeat the procedure if $\overline{T}_k$ too high
  complexity, or does not lie in between the $k$-level and
  $(1 + \eps)k$-level of $\A(\nabla)$. In expectation this requires a
  constant number of restarts. The lemma follows.
\end{proof}

As an $\beta$-approximate $k$-level is also an $\eps$-approximate
$k$-level for all $\eps > \beta$, we can construct
$\frac{1}{2}$-approximate $k$-levels instead of $\eps$-approximate
$k$-levels when $\eps > \frac{1}{2}$.  This results in the following
corollary.

\begin{corollary}
  \label{cor:approx_level_Linfty}
    Let $\nabla$ be a set of $n$ pyramids.
    Let $1 \leq k \leq n$ and $\eps > 0$ be parameters.
    Let $\beta = \min\{ \eps, 1/2 \}$.
    In
    $
        O\left( \frac{n}{\beta^2}\lambda_4\left( \frac{\log n}{\beta^2} \right)\log^3 n \right)
    $
    expected time, we can build an $\eps$-approximate $k$-level of
    $\A(\nabla)$ of
    complexity $O\left( \frac{n}{\beta^5 k} \log^2 n \right)$.
\end{corollary}

\subparagraph{The data structure.} Set $\alpha = 1 - \sqrt{1 - \eps}$,
and let $\nabla_c$ be the set of pyramids in $\nabla$ that have color
$c$.  For each set $\nabla_c$, we build
$\left( \frac{\alpha}{1 - \alpha} \right)$-approximate
$\left( \frac{1}{1 - \alpha} \right)^i$-levels $L_{c, i}$ of
$\A(\nabla_c)$, for
$i = 0, \dots, \log_{\frac{1}{1 - \alpha}} |\nabla_c|$, and define
$\LL_i$ as the $0$-level of all surfaces $L_{c,i}$ over all colors
$c$. Using the exact same argument as in
Lemma~\ref{lem:approximate_mode} it then follows that if $\hat{q}$
lies in between $\LL_i$ and $\LL_{i+1}$, the vertical ray with top
endpoint $\hat{q}$ intersects with at least $(1-\eps)f^*$ pyramids of
the same color as the pyramid defining $\LL_i$ directly below
$\hat{q}$. Hence we can solve $\eps$-approximate range mode queries by
binary searching over $0,\dots,\log_{\frac{1}{1 - \alpha}} n$. To test if
$\hat{q}$ lies above or below $\LL_i$ can again decompose each
$L_{c,i}$ into triangles, and store them in the vertical ray shooting
data structure of Lemma~\ref{lem:ray_shooting}. We therefore obtain
the following result:

\begin{theorem}
  \label{thm:approx_range_mode_squares}
  Let $P$ be a set of $n$ points in $\R^2$ and $\eps \in (0, 1)$ a parameter.
  In expected $O\left( \frac{n}{\eps^2}\log^2 n\left( \lambda_4\left( \frac{\log n}{\eps^2}
      \right)\log n\log_{\frac{1}{1-\eps}}
      n + \frac{1}{\eps^4}\log^4\frac{n}{\eps}\right)
  \right)$ time, we can build a data structure of size
  $O\left( \frac{n}{\eps^6}\log^2 n\log^2 \left(
      \frac{n}{\eps}\right)\right)$, that can answer
  $\eps$-approximate range mode queries with a square query region in
  $O\left(\frac{n^{1/2}}{\eps^{5/2}}\log n\log^3
    \frac{n}{\eps}\log\log_{\frac{1}{1-\eps}} n\right)$ expected time.
\end{theorem}

\begin{proof}
  As argued above our data structure correctly answers queries, so all
  that remains is to analyze the space usage, and the preprocessing
  and query times.

  \subparagraph{Space usage and preprocessing time.} Let
  $\beta = \min \left\{ \frac{\alpha}{1 - \alpha}, \frac{1}{2}
  \right\}$. By Corollary~\ref{cor:approx_level_Linfty} the total
  complexity of the approximate levels is
  \begin{align*}
    O\left( \sum_c \sum_{i = 0}^{\log_{\frac{1}{1 - \alpha}} |\nabla_c|} \frac{|\nabla_c| \log^2 |\nabla_c|}{\beta^5 \cdot \left( \frac{1}{1 - \alpha} \right)^i} \right)
    &= O\left( \frac{n \log^2 n}{\beta^5} \cdot \sum_{i =
        0}^{\log_{\frac{1}{1 - \alpha}} n} \frac{1}{\left( \frac{1}{1
            - \alpha} \right)^i} \right)\\
    &= O\left( \frac{n \log^2 n}{\beta^5} \cdot \sum_{i =
        0}^{\log_{\frac{1}{1 - \alpha}} n} \left(1 - \alpha \right)^i \right)
    = O\left( \frac{n \log^2 n}{\alpha \cdot \beta^5} \right).
  \end{align*}
  Where this last step again uses that
  $\sum_{i=0}^\infty (1-\alpha)^i$ is a geometric series converging to
  $1/\alpha$. Also as before, we have
  $\eps / 2 \leq \alpha = 1 - \sqrt{1 - \eps} \leq \eps$, and thus
  $\alpha = \Theta(\eps)$. Furthermore, we have
  $\beta = \min \left\{ \frac{\alpha}{1 - \alpha}, \frac{1}{2}
  \right\} \geq \min\{ \alpha, 1 / 2 \}$, and thus
  $\beta = \Omega(\eps)$. The total complexity of all approximate
  levels is thus
  $O\left( \frac{n}{\alpha \cdot \beta^5}\log^2 n \right)= O\left(
    \frac{n}{\eps^6}\log^2 n \right)$. Since for any $k$ constructing
  an $(\frac{\alpha}{1-\alpha})$-approximate $k$-level of $m$ pyramids
  takes expected
  $O\left(\frac{m}{\beta^2}\lambda_4\left( \frac{\log m}{\beta^2}
    \right)\log^3 m \right)$ time
  (Corollary~\ref{cor:approx_level_Linfty}), constructing all of them
  takes expected
  \begin{align*}
    O\left( \sum_c \sum_{i = 0}^{\log_{\frac{1}{1 - \alpha}}
        |\nabla_c|}
      \frac{|\nabla_c|}{\beta^2}\lambda_4\left( \frac{\log
          |\nabla_c|}{\beta^2} \right)\log^3 |\nabla_c| \right)
    &= O\left(\frac{n}{\beta^2}\lambda_4\left( \frac{\log n}{\beta^2}
      \right)\log^3 n\log_{\frac{1}{1 - \alpha}} n \right)\\
    &= O\left(\frac{n}{\eps^2}\lambda_4\left( \frac{\log n}{\eps^2}
    \right)\log^3 n\log_{\frac{1}{1 - \eps}} n \right)
  \end{align*}
  time in total.

  By Lemma~\ref{lem:ray_shooting} the ray shooting data structures
  can be built in $O\left( \frac{n}{\eps^6}\log^2 n\log^4 \left(
    \frac{n}{\eps^6}\right)\right)$ time, and they use
  $O\left(
    \frac{n}{\eps^6}\log^2 n\log^2 \left(
      \frac{n}{\eps}\right) \log \log \left( 
        \frac{n}{\eps} \right)\right)$ space.
  It now follows that the the total expected
  construction time is
  $O\left(\frac{n}{\eps^2}\lambda_4\left( \frac{\log n}{\eps^2}
    \right)\log^3 n\log_{\frac{1}{1 - \eps}} n +
    \frac{n}{\eps^6}\log^2 n\log^4 \left(
      \frac{n}{\eps}\right)\right) =
  O\left( \frac{n}{\eps^2}\log^2 n\left( \lambda_4\left( \frac{\log n}{\eps^2}
    \right)\log n\log_{\frac{1}{1-\eps}}
      n + \frac{1}{\eps^4}\log^4\frac{n}{\eps}\right)
  \right)
  $, and the total space used by our data structure is
  $O\left( \frac{n}{\eps^6}\log^2 n\log^2 \left(
      \frac{n}{\eps}\right) \log \log \left( 
        \frac{n}{\eps} \right)\right)$.

  \subparagraph{Query time.} The query algorithm binary searches over
  the range $0, \dots, \log_{\frac{1}{1 - \alpha}} n$, and queries the ray
  shooting structure in every step. The maximum complexity of $\LL_i$
  for any $i$ is 
  $O\left(\sum_c \frac{|\nabla_c| \log^2 |\nabla_c|}{\beta^5 \cdot \left( 
    \frac{1}{1 - \alpha} \right)^i}\right) = 
    O(\frac{n}{\eps^5} \log^2 n)$.
  Thus by Lemma~\ref{lem:ray_shooting} a ray shooting query takes
  $O\left(\frac{n^{1/2}}{\eps^{5/2}} \log n \log^2
    \frac{n}{\eps}\right)$ expected time. Hence, the total expected
  query time is \\
  $O\left(\frac{n^{1/2}}{\eps^{5/2}}\log n\log^2
    \frac{n}{\eps}\log\log_{\frac{1}{1-\eps}} n\right)$.
\end{proof}

\subparagraph{Approximate chromatic \kNN queries.} We use the the data
structure from Theorem~\ref{thm:2d_finding_query_range_L_infty_nlogn}
to compute the smallest query range $\Dk(q)$ of radius $r$ that
contains $\kNN(q)$, and then use the data structure from
Theorem~\ref{thm:approx_range_mode_squares} to obtain the following
result:

\begin{theorem}
  \label{thm:approximate_chromatic_kNN_L_infty}
  Let $P$ be a set of $n$ points in $\R^2$ and $\eps \in (0, 1)$ a parameter.
  In expected \linebreak
  $O\left( \frac{n}{\eps^2}\log^2 n\left( \lambda_4\left( \frac{\log
          n}{\eps^2} \right)\log n\log_{\frac{1}{1-\eps}} n +
      \frac{1}{\eps^4}\log^4\frac{n}{\eps}\right) \right)$ time, we
  can build a data structure of size
  $O\left( \frac{n}{\eps^6}\log^2 n\log^2 \left(
      \frac{n}{\eps}\right)\right)$, that can answer
  $\eps$-approximate chromatic \kNN queries on $P$ in the $L_\infty$ metric in
  $O\left(\frac{n^{1/2}}{\eps^{5/2}}\log n\log^3
    \frac{n}{\eps}\log\log_{\frac{1}{1-\eps}} n\right)$ expected time.
\end{theorem}

\section{Extensions}
\label{sec:Extensions}

In this section, we sketch how to generalize our data structures that
answer queries exactly, i.e. the results of
Sections~\ref{sec:finding_the_range_2d} and~\ref{sec:range_mode_2d},
to work in higher dimensions (Section~\ref{sub:higher_dimensions}),
i.e. $d \geq 2$, and under every $L_m$ metric where
$m$ is a constant (Section~\ref{sub:general_metrics}). Unfortunately,
the complexity of shallow cuttings and approximate levels in
dimensions $d+1 \geq 4$ is rather high. Therefore, our approach for answering
approximate \kNN queries does not easily generalize to higher
dimensions.

\subsection{Higher dimensions}
\label{sub:higher_dimensions}

\subparagraph*{Range finding queries.}
For finding $\Dk_{\infty}(q)$ in $\R^2$, recall that the points in $P$ 
are stored in two balanced binary search trees, sorted once on 
$x$-coordinate and once on $y$-coordinate. Aside from this, we store 
$P$ in a data structure for two-dimesional orthogonal range counting 
queries. The $x$- and $y$-coordinates define radii $r$, on which we 
can binary search, combined with range counting, to find $\Dk_{\infty}(q)$. 
In $d$ dimensions, the data structure is similar. We store $P$ in $d$ 
balanced binary search trees $T_i$, with $T_i$ storing the ${x_i}^{th}$ 
coordinates of the points in $P$. We also store $P$ in a data 
structure for $d$-dimensional orthogonal range counting~\cite{DBLP:journals/cacm/Bentley80, willard85orthogonal}. The query 
algorithm can then easily be extended to work in this higher 
dimensional data structure, where we perform $O(d \log n)$ range 
counting queries. We thus obtain:

\begin{theorem}
  Let $P$ be a set of $n$ points in $\R^d$. In $O(n \log^{d-1} n)$ 
  time, we can build a data structure of $O(n \log^{d-1} n)$ size, 
  that can report $\Dk_{\infty}(q)$ in $O(\log^d n)$ time.
\end{theorem}

We can generalize the linear-space orthogonal range counting data structure
from Section~\ref{sub:finding_the_k-nearest_neighbors_under_the_L_infty_metric}
to higher dimensions. This results in a linear-space data structure that answers
queries in $O(n^\delta)$ time. We get the following result:

\begin{theorem}
  Let $P$ be a set of $n$ points in $\R^d$, with $d \geq 2$. Let $\delta > 0$
  be an arbitrarily small constant. In $O(n \log n)$ 
  time, we can build a data structure of $O(n)$ size, 
  that can report $\Dk_{\infty}(q)$ in $O(n^\delta)$ time.
\end{theorem}

To find $\Dk_2(q)$ in a higher dimensional setting, we can simply use 
the data structure from Section~\ref{subsub:deter_finding_disk}. The 
main part of the data structure, the partition tree of Agarwal 
\etal~\cite{agarwal13semialgebraic}, can readily be built on higher 
dimensional data, and its structure remains the same; each node 
corresponds to a cell in $\R^d$, containing $O(n / r)$ points of $P$ 
each. (Recall that $r$ is a parameter for the data structure.) A node 
has $O(r)$ children, and a query range (ball) crosses at most $O(r^{1-1/d})$ cells of these children. This gives the following result for 
when $P$ is in $D_0$-general position. We lift this assumption using 
the technique from Section~\ref{subsub:deter_finding_disk} for 
handling arbitrary data sets.

\begin{theorem}
  Let $P$ be a set of $n$ points in $\R^d$, with $d \geq 2$. Let $\delta > 0$ be an 
  arbitrarily small constant. In $O(n^{1+\delta})$ expected time, we 
  can build a data structure of $O(n)$ size, that can report $\Dk_2(q)$ in $O(n^{1-1/d} \polylog n)$ time.
\end{theorem}

\subparagraph*{Range mode queries.}
For range mode queries under both the $L_2$ and $L_\infty$ metric, we 
used cuttings in three-dimensional space, built on a set of planes (or 
in the case of the $L_\infty$ metric, pyramids, for which the cutting 
is constructed on planes first). We use the result of Chazelle~\cite{DBLP:journals/dcg/Chazelle93a} on constructing optimal $O(r^{d+1})
$-size $\frac{1}{r}$-cuttings (in $\R^{d+1}$) in $O(nr^d)$ time. 

When using the $L_2$ metric, we use the result of Agarwal~\etal~\cite{agarwal13semialgebraic} on semialgebraic range searching, for 
reporting the conflict list of a simplex and performing level queries. 
Note that like in Section~\ref{sec:rm2_L2}, these queries can be 
performed in $\R^d$. Setting $r = n^{1/(d+1)}$, we can construct a 
linear-size data structure that answers 
queries though $O(n / r)$ semialgebraic range counting queries on 
disjoint subsets of $P$. The query time is thus 
$O(n / r^{1-1/d} \polylog n) = O(n^{1-(d-1)/(d^2+d)} \polylog n)$.
For preprocessing, we generalize the algorithm of Section~\ref{sec:rm2_prep}
to higher dimensions. Note that the graph traversal method used in
the algorithm works regardless of dimension. We only need to use
higher-dimensional range-searching data structures to update the color
frequencies when traversing an edge. We need
$O(n^{1-1/d} \polylog n + n/r)$ time per edge, totalling $O(n^{1-1/d} r^{d+1} \polylog n + nr^d) = O(n^{2-1/d} \polylog n + n^{1+d/(d+1)})$
time for the whole cutting. Note that this dominates the expected 
preprocessing time for the range-searching data structure, and thus
bounds the expected preprocessing time. We obtain:

\begin{theorem}
  Let $P$ be a set of $n$ colored points in $\R^d$, with $d \geq 2$. In expected
  $O(n^{2-1/d} \polylog n + n^{1+d/(d+1)})$ time, we can build a data 
  structure of $O(n)$ size, that 
  reports the mode color among the points in a query ball in $O(n^{1-(d-1)/(d^2+d)} \polylog n)$ time.
\end{theorem}

When using the $L_\infty$ metric, we use range trees~\cite{DBLP:journals/cacm/Bentley80, willard85orthogonal} for reporting the 
conflict list of a simplex and performing level queries. We now set $r = (n \log^{d-1} n)^{1/(d+1)}$, to match the size of the cutting
with the $O(n \log^{d-1} n)$ space 
used by the range tree. Querying the range tree takes $O(\log^{d-1} n + k')$ time, where $k'$ is the output size. 
By again generalizing the faster preprocessing algorithm, we
get a preprocessing time of $O(n \log^{d-1} n + r^{d+1} \log^{d-1} n + nr^d) = O(n^{1+d/(d+1)} \log^{(d^2-d)/(d+1)} n)$.
The query time is $O(n/r \cdot \log^{d-1} n) = O((n \log^{d-1} n)^{1-1/(d+1)})$.
This gives the following result:
  
\begin{theorem}
  Let $P$ be a set of $n$ colored points in $\R^d$, with $d \geq 2$. In \\
  $O(n^{1+d/(d+1)} \log^{(d^2-d)(d+1)} n)$
  time, we can build a data structure of $O(n \log^{d-1} n)$ size, that 
  reports the mode color among the points in a query hypercube in $O((n \log^{d-1} n)^{1-1/(d+1)})$ time.
\end{theorem}

We can create the range trees with $n^\delta$ fan-out, to get a linear-size data structure (see Section~\ref{sub:finding_the_k-nearest_neighbors_under_the_L_infty_metric}), that answers queries in $O(n^\delta + k')$ time. 
We then set $r = n^{1/(d+1)}$ to obtain a linear-size cutting.
The preprocessing time is $O(n \log n + n^\delta r^{d+1} + nr^d) = O(n^{1+d/(d+1)})$, and the query time is
$O(n/r \cdot n^\delta) = O(n^{1-1/(d+1)+\delta})$. We summarize:

\begin{theorem}
  Let $P$ be a set of $n$ colored points in $\R^d$, with $d \geq 2$.
  Let $\delta > 0$ be an arbitrarily small constant. In $O(n^{1+d/(d+1)})$
  time, we can build a data structure of $O(n \log n)$ size, that 
  reports the mode color among the points in a query hypercube in $O(n^{1-1/(d+1)+\delta})$ time.
\end{theorem}

Combined with the results on higher dimensional range finding queries, we obtain the following results:

\begin{theorem}
  Let $P$ be a set of $n$ points in $\R^d$, with $d \geq 2$. 
  Let $\delta > 0$ be an arbitrarily small constant. There is a linear space
  data structure that answers chromatic $k$-NN queries on $P$ with
  respect to the $L_2$ metric. Building the data structure takes expected
  $O(n^{2-1/d} \polylog n + n^{1+d/(d+1)})$ time, and queries take $O(n^{1-(d-1)/(d^2+d)} \polylog n)$ time.
\end{theorem}

\begin{theorem}
  Let $P$ be a set of $n$ points in $\R^d$, with $d \geq 2$. 
  Let $\delta > 0$ be an arbitrarily small constant. There is a linear space
  data structure that answers chromatic $k$-NN queries on $P$ with
  respect to the $L_\infty$ metric. Building the data structure takes 
  $O(n^{1+d/(d+1)})$ worst-case time, and queries take $O(n^{1-1/(d+1)+\delta})$ time.
  We can decrease the query time to $O((n \log^{d-1} n)^{1-1/(d+1)})$ using
  $O(n \log^{d-1} n)$ space and $O(n^{1+d/(d+1)} \log^{(d^2-d)/(d+1)} n)$
  preprocessing time.
\end{theorem}

\subsection{General $L_m$ metrics}
\label{sub:general_metrics}

Our results, including the higher dimensional results sketched above, 
can be extended to work for general $L_m$ metrics as well, for $m$ a 
constant. 
For this, we use the range finding data structure for finding $\Dk_2(q)$, 
so the variant used under the $L_2$ metric, and for range mode 
queries, we augment the solution for the $L_\infty$ metric.

\subparagraph*{Range finding queries.}
The data structure for finding $\Dk_2(q)$ is easily extended for 
finding $\Dk_m(q)$, if $m$ is a constant. All that is required is to build 
the partition tree for the family of ranges that have the shape of the 
unit ball under the $L_m$ metric. Because for constant $m$, this unit
ball can described using a constant number of polynomials, each with
degree $m = O(1)$, this can be done using the same 
result of Agarwal~\etal~\cite{agarwal13semialgebraic} that is used for 
the $L_2$ metric. We get the following result:

\begin{theorem}
  Let $P$ be a set of $n$ points in $\R^d$, with $d \geq 2$. Let $\delta > 0$ be an 
  arbitrarily small constant. In $O(n^{1+\delta})$ expected time, we 
  can build a data structure of $O(n)$ size, that can report $\Dk_m(q)$ in $O(n^{1-1/d} \polylog n)$ time.
\end{theorem}

\subparagraph*{Range mode queries.}
To answer range mode queries under the $L_m$ metric, we follow the
approach of Section~\ref{sec:rm2_Linfty}. Let $B_m(q, r) = \{p \in \R^d
\mid L_m(p, q) \leq r\}$ be the metric ball in $\R^d$, centered at $q$
and with radius $r$. Let $d' = d + 1$, and let $S_m = \{L_m(p, \cdot)
\mid p \in P\}$ be the set of graphs in $\R^{d'}$, of the distance
functions from points in $P$. We then have that a point $q$ lies in a
range $B_m(p, r)$ if and only if the point $\hat{q} = (q_1, \dots,
q_d, r)$ lies vertically (with respect to the $x_{d+1}$ axis) above
$L_m(p, \cdot)$. Thus, the $k$ nearest neighbors of $q$ correspond to
those surfaces of $S_m$ below $\hat{q}$. We now need a
$\frac{1}{r}$-cutting for the arrangement $\A(S_m)$ of $S_m$. Here we
use the generalized definition of a $\frac{1}{r}$-cutting of Agarwal
and \Matousek~\cite{agarwal94range_searc_semial_sets}. They define a
$\frac{1}{r}$-cutting as a set of \emph{elementary cells}; closed
regions with constant description complexity, that have disjoint
interiors, such that the interior of each elementary cell intersects
at most $n / r$ surfaces of $S_m$. They also prove that there is a
$\frac{1}{r}$-cutting of size $O(T(r \log r))$, where $T(r)$ is the
size of the \emph{elementary cell decomposition} of a set of $r$
surfaces. Using the random sampling result of Clarkson~\cite[Corollary
4.2]{clarkson1987random}, we get the following result on constructing
such a $\frac{1}{r}$-cutting:

\begin{lemma}
  Let $r \in [1, n]$ be a parameter and $R_m \subset S_m$ be a random subset of size $r$. Then with probability at least $1/2$, the vertical decomposition of $R_m$ will form an $O\left( \frac{\log r}{r} \right)$-cutting for $S_m$.
\end{lemma}

Note that slightly smaller cuttings are known to exists as
well~\cite{berg95cutting}, however, since this does not impact our
asymptotic bounds, we use the slightly simpler result from Clarkson
instead. Using the results of Chazelle
\etal~\cite{chazelle91stratification} and
Koltun~\cite{koltun04vertical_decomposition}, we obtain the following
result on constructing cuttings and point location structures on them:

\begin{lemma}
  Let $r \in [1, n]$ be a parameter and $\delta > 0$ be an arbitrarily small constant. We can construct a cutting $\Xi$ of $S_m$, and with probability at least $1/2$, $\Xi$ will be an $O\left( \frac{\log r}{r} \right)$-cutting of $S_m$.
  If $d' = 3$, then we can construct $\Xi$, together with a point location data structure on $\Xi$, in $O(r^{4+\delta})$ expected time. The space used is $O(r^{4+\delta})$, and point location queries can be answered in $O(\log r)$ time.
  If $d' \geq 4$, then we can construct $\Xi$ and the point location data structure in $O(r^{2d'-4+\delta})$ worst-case time. The space used is $O(r^{2d'-4+\delta})$.
\end{lemma}

We now construct two of these cuttings $\Xi_1$ and $\Xi_2$ on $S_m$.
The remainder of the data structure is the same as for that of 
Section~\ref{sec:rm2_Linfty}, except that the orthogonal range searching 
data structures are replaced with the semialgebraic range searching 
data structures of Agarwal~\etal~\cite{agarwal13semialgebraic}. The 
main idea behind the query algorithm remains unaltered as well. Given a 
query range $B(q, r)$, we locate the regions $\Omega_1$ and $\Omega_2$ 
containing $\hat{q}$ in $\Xi_1$ and $\Xi_2$, respectively. With range 
counting, we can determine which region has the smallest conflict list. 
In expectation, this conflict list will have size $O(n\log r / r)$. We 
then continue the query only with this region. The remainder of the 
query algorithm is the same as for the $L_\infty$ metric in $\R^2$. 
Setting $r = n^{1/(4+\delta)}$ if $d' = 3$, and $r = n^{1/(2d'-4+\delta)}$ 
when $d' > 4$, we get a linear-size data structure with an expected query 
time of $O( n^{1-1/(12+3\delta)} \polylog n )$ when $d' = 3$,
and $O( n^{1-1/(d'(2d'-4+\delta))} \polylog n )$ when 
$d' \geq 4$. Note that we can use the same preprocessing algorithm as for
when the $L_2$ metric is used (see Section~\ref{sec:rm2_prep}), using
semialgebraic range searching in combination with the graph traversal of
the vertices of $\Xi_1$ and $\Xi_2$. After constructing the semialgebraic
range searching data structure, we can then compute the colors stored
in the cells of $\Xi_1$ and $\Xi_2$ in expected
$O( n^{1-1/(d'-1)} (|\Xi_1| + |\Xi_2|) \polylog n + nr^{d'-1} \log^{d'} r )$ time,
where the expectation is taken with respect to the total size of the
conflict lists.
Note that because both cuttings have linear size, this preprocessing time 
dominates the expected preprocessing time of the range searching data 
structure, giving a total expected preprocessing time of
$O(n^{\frac{2}{4+\delta}} \polylog n)$ when $d' = 3$, and $O((n^{2-1/(d'-1)} + n^{1+(d'-1)/(2d'-4+\delta)}) \polylog n)$ when $d' \geq 4$.
Now recall that $d' = d+1$. We summarize:

\begin{theorem}
  Let $P$ be a set of $n$ colored points in $\R^d$, with $d \geq 2$. Let $\delta > 0$
  be an arbitrarily small constant. In expected
  $O( n^{2/(4+\delta)} \polylog n )$ time when $d=2$,
  and expected $O( ( n^{2-1/d} + n^{1+d/(2d-2+\delta)} ) \polylog n )$ time when $d \geq 3$, we can build 
  a data structure of $O(n)$ size, that 
  reports the mode color among the points in a query ball under metric 
  $L_m$, for $m = O(1)$, in expected $O( n^{1-1/(12+3\delta)} \polylog n )$ 
  time when $d = 2$, and expected $O( n^{1-1/((d+1)(2d-2+\delta))} \polylog n )$ time when $d \geq 3$.
\end{theorem}

Combined with the result on finding $\Dk_m(q)$, we get the following result for chromatic \kNN queries under the $L_m$ metrics:

\begin{theorem}
  Let $P$ be a set of $n$ points in $\R^d$, with $d \geq 2$. Let $\delta > 0$
  be an arbitrarily small constant. There is a linear space
  data structure that answers chromatic $k$-NN queries on $P$ with
  respect to the $L_m$ metric, for $m = O(1)$. Building the data 
  structure takes expected
  $O( n^{2/(4+\delta)} \polylog n )$ time when $d=2$,
  and expected $O(( n^{2-1/d} + n^{1+d/(2d-2+\delta)} ) \polylog n )$ time when $d \geq 3$. 
  Queries take $O( n^{1-1/(12+3\delta)} \polylog n )$ time when $d = 2$, and 
  $O( n^{1-1/(d+1)(2d-2+\delta)} \polylog n )$ time when 
  $d \geq 3$.
\end{theorem}

\section{Concluding Remarks}
\label{sec:Concluding_Remarks}

We presented the first data structures for the chromatic \kNN problem
with query times that depend only on the number of stored points. Our
two step approach essentially reduces the problem to efficiently
answering range mode queries. The main open question is how to answer
such queries efficiently. Since it is unlikely that we can answer such
queries in $\Omega(n^{1/2})$ time (using only near-linear space), it
is also particularly interesting to consider further improvements to
the $\eps$-approximate query data structures. For the Euclidean
distance, finding the query range may now be the dominant factor
(depending on the choice of $\eps$). One option is to report a range
that contains only approximately $k$ points. However, this further
complicates the analysis. For the $L_\infty$ distance it
may also be possible to reduce the space usage by using a different
method for computing the approximate levels (e.g. using the results by
Agarwal~\etal~\cite{DBLP:journals/siamcomp/AgarwalES99} or the recent
results of Liu~\cite{liu20nearl_optim_planar_neares_neigh}).

\bibliography{bibliography}

\begin{thebibliography}{10}

\bibitem{DBLP:journals/siamcomp/AgarwalES99}
Pankaj~K. Agarwal, Alon Efrat, and Micha Sharir.
\newblock Vertical decomposition of shallow levels in 3-dimensional
  arrangements and its applications.
\newblock {\em {SIAM} Journal on Computing}, 29:912--953, 1999.
\newblock \href {https://doi.org/10.1137/S0097539795295936}
  {\path{doi:10.1137/S0097539795295936}}.

\bibitem{agarwal94range_searc_semial_sets}
Pankaj~K. Agarwal and Jir{\'{\i}} Matousek.
\newblock On range searching with semialgebraic sets.
\newblock {\em Discrete \& Computational Geometry}, 11:393--418, 1994.
\newblock \href {https://doi.org/10.1007/BF02574015}
  {\path{doi:10.1007/BF02574015}}.

\bibitem{agarwal13semialgebraic}
Pankaj~K. Agarwal, Jir{\'{\i}} Matousek, and Micha Sharir.
\newblock On range searching with semialgebraic sets. {II}.
\newblock {\em {SIAM} Journal on Computing}, 42(6):2039--2062, 2013.
\newblock \href {https://doi.org/10.1137/120890855}
  {\path{doi:10.1137/120890855}}.

\bibitem{aggarwal2014data}
Charu~C. Aggarwal.
\newblock {\em Data classification: algorithms and applications}.
\newblock CRC press, 2014.

\bibitem{arya2012halfspace}
Sunil Arya, David~M. Mount, and Jian Xia.
\newblock Tight lower bounds for halfspace range searching.
\newblock {\em Discrete \& Computational Geometry}, 47(4):711--730, 2012.
\newblock \href {https://doi.org/10.1007/s00454-012-9412-x}
  {\path{doi:10.1007/s00454-012-9412-x}}.

\bibitem{DBLP:conf/focs/BansalW09}
Nikhil Bansal and Ryan Williams.
\newblock Regularity lemmas and combinatorial algorithms.
\newblock In {\em 50th Annual {IEEE} Symposium on Foundations of Computer
  Science}, pages 745--754, 2009.
\newblock \href {https://doi.org/10.1109/FOCS.2009.76}
  {\path{doi:10.1109/FOCS.2009.76}}.

\bibitem{DBLP:journals/cacm/Bentley80}
Jon~Louis Bentley.
\newblock Multidimensional divide-and-conquer.
\newblock {\em Communications of the {ACM}}, 23(4):214--229, 1980.
\newblock \href {https://doi.org/10.1145/358841.358850}
  {\path{doi:10.1145/358841.358850}}.

\bibitem{bohler15voron}
Cecilia Bohler, Panagiotis Cheilaris, Rolf Klein, Chih-Hung Liu, Evanthia
  Papadopoulou, and Maksym Zavershynskyi.
\newblock On the complexity of higher order abstract voronoi diagrams.
\newblock {\em Computational Geometry}, 48(8):539--551, 2015.
\newblock \href {https://doi.org/10.1016/j.comgeo.2015.04.008}
  {\path{doi:10.1016/j.comgeo.2015.04.008}}.

\bibitem{DBLP:conf/stacs/BoseKMT05}
Prosenjit Bose, Evangelos Kranakis, Pat Morin, and Yihui Tang.
\newblock Approximate range mode and range median queries.
\newblock In {\em {STACS}}, pages 377--388, 2005.
\newblock \href {https://doi.org/10.1007/978-3-540-31856-9\_31}
  {\path{doi:10.1007/978-3-540-31856-9\_31}}.

\bibitem{chan2012optimal}
Timothy~M. Chan.
\newblock Optimal partition trees.
\newblock {\em Discrete \& Computational Geometry}, 47(4):661--690, 2012.
\newblock \href {https://doi.org/10.1007/s00454-012-9410-z}
  {\path{doi:10.1007/s00454-012-9410-z}}.

\bibitem{DBLP:conf/soda/Chan15}
Timothy~M. Chan.
\newblock Speeding up the four russians algorithm by about one more logarithmic
  factor.
\newblock In {\em Proceedings of the Twenty-Sixth Annual {ACM-SIAM} Symposium
  on Discrete Algorithms}, pages 212--217, 2015.
\newblock \href {https://doi.org/10.5555/2722129.2722145}
  {\path{doi:10.5555/2722129.2722145}}.

\bibitem{chan14linear_space_data_struc_range}
Timothy~M. Chan, Stephane Durocher, Kasper~Green Larsen, Jason Morrison, and
  Bryan~T. Wilkinson.
\newblock Linear-space data structures for range mode query in arrays.
\newblock {\em Theory of Computing Systems}, 55:719--741, 2014.
\newblock \href {https://doi.org/10.1007/s00224-013-9455-2}
  {\path{doi:10.1007/s00224-013-9455-2}}.

\bibitem{chan20furth_resul_color_range_searc}
Timothy~M. Chan, Qizheng He, and Yakov Nekrich.
\newblock {Further Results on Colored Range Searching}.
\newblock In {\em 36th International Symposium on Computational Geometry},
  pages 28:1--28:15, 2020.
\newblock \href {https://doi.org/10.4230/LIPIcs.SoCG.2020.28}
  {\path{doi:10.4230/LIPIcs.SoCG.2020.28}}.

\bibitem{DBLP:journals/dcg/Chazelle93a}
Bernard Chazelle.
\newblock Cutting hyperplanes for divide-and-conquer.
\newblock {\em Discrete \& Computational Geometry}, 9:145--158, 1993.
\newblock \href {https://doi.org/10.1007/BF02189314}
  {\path{doi:10.1007/BF02189314}}.

\bibitem{chazelle91stratification}
Bernard Chazelle, Herbert Edelsbrunner, Leonidas~J. Guibas, and Micha Sharir.
\newblock A singly exponential stratification scheme for real semi-algebraic
  varieties and its applications.
\newblock {\em Theoretical Computer Science}, 84(1):77--105, 1991.
\newblock \href {https://doi.org/10.1016/0304-3975(91)90261-Y}
  {\path{doi:10.1016/0304-3975(91)90261-Y}}.

\bibitem{clarkson1987random}
Kenneth~L. Clarkson.
\newblock New applications of random sampling in computational geometry.
\newblock {\em Discrete \& Computational Geometry}, 2:195--222, 1987.
\newblock \href {https://doi.org/10.1007/BF02187879}
  {\path{doi:10.1007/BF02187879}}.

\bibitem{cover1967nearest}
Thomas Cover and Peter Hart.
\newblock Nearest neighbor pattern classification.
\newblock {\em IEEE transactions on information theory}, 13(1):21--27, 1967.
\newblock \href {https://doi.org/10.1109/TIT.1967.1053964}
  {\path{doi:10.1109/TIT.1967.1053964}}.

\bibitem{DBLP:books/sp/Berg93}
Mark de~Berg.
\newblock {\em Ray Shooting, Depth Orders and Hidden Surface Removal}, volume
  703 of {\em Lecture Notes in Computer Science}.
\newblock Springer, 1993.
\newblock \href {https://doi.org/10.1007/BFb0029813}
  {\path{doi:10.1007/BFb0029813}}.

\bibitem{berg95cutting}
Mark de~Berg and Otfried Schwarzkopf.
\newblock Cuttings and applications.
\newblock {\em International Journal of Computational Geometry \&
  Applications}, 5(4):343--355, 1995.
\newblock \href {https://doi.org/10.1142/S0218195995000210}
  {\path{doi:10.1142/S0218195995000210}}.

\bibitem{DBLP:journals/dcg/Edelsbrunner89}
Herbert Edelsbrunner.
\newblock The upper envelope of piecewise linear functions: Tight bounds on the
  number of faces.
\newblock {\em Discrete \& Computational Geometry}, 4:337--343, 1989.
\newblock \href {https://doi.org/10.1007/BF02187734}
  {\path{doi:10.1007/BF02187734}}.

\bibitem{DBLP:journals/tog/EdelsbrunnerM90}
Herbert Edelsbrunner and Ernst~P. M{\"{u}}cke.
\newblock Simulation of simplicity: a technique to cope with degenerate cases
  in geometric algorithms.
\newblock {\em {ACM} Transactions on Graphics}, 9:66--104, 1990.
\newblock \href {https://doi.org/10.1145/77635.77639}
  {\path{doi:10.1145/77635.77639}}.

\bibitem{DBLP:journals/siamcomp/EmirisC95}
Ioannis~Z. Emiris and John~F. Canny.
\newblock A general approach to removing degeneracies.
\newblock {\em {SIAM} Journal on Computing}, 24:650--664, 1995.
\newblock \href {https://doi.org/10.1137/S0097539792235918}
  {\path{doi:10.1137/S0097539792235918}}.

\bibitem{DBLP:conf/focs/FischerM71}
Michael~J. Fischer and Albert~R. Meyer.
\newblock Boolean matrix multiplication and transitive closure.
\newblock In {\em 12th Annual Symposium on Switching and Automata Theory},
  pages 129--131, 1971.
\newblock \href {https://doi.org/10.1109/SWAT.1971.4}
  {\path{doi:10.1109/SWAT.1971.4}}.

\bibitem{friedman77logarithmic_classification}
Jerome~H. Friedman, Jon~Louis Bentley, and Raphael~A. Finkel.
\newblock An algorithm for finding best matches in logarithmic expected time.
\newblock {\em {ACM} Transactions on Mathematical Software}, 3(3):209--226,
  1977.
\newblock \href {https://doi.org/10.1145/355744.355745}
  {\path{doi:10.1145/355744.355745}}.

\bibitem{Har-Peled2017}
Sariel Har-Peled, Haim Kaplan, and Micha Sharir.
\newblock Approximating the k-level in three-dimensional plane arrangements.
\newblock In {\em A Journey Through Discrete Mathematics: A Tribute to
  Ji{\v{r}}{\'i} Matou{\v{s}}ek}, pages 467--503. 2017.
\newblock \href {https://doi.org/10.1007/978-3-319-44479-6_19}
  {\path{doi:10.1007/978-3-319-44479-6_19}}.

\bibitem{henley96neares_neigh_class_asses_consum_credit_risk}
W.~E. Henley and D.~J. Hand.
\newblock A $k$-nearest-neighbour classifier for assessing consumer credit
  risk.
\newblock {\em Journal of the Royal Statistical Society. Series D (The
  Statistician)}, 45(1):77--95, 1996.
\newblock \href {https://doi.org/10.2307/2348414} {\path{doi:10.2307/2348414}}.

\bibitem{DBLP:journals/dcg/KaplanMRSS20}
Haim Kaplan, Wolfgang Mulzer, Liam Roditty, Paul Seiferth, and Micha Sharir.
\newblock Dynamic planar voronoi diagrams for general distance functions and
  their algorithmic applications.
\newblock {\em Discrete \& Computational Geometry}, 64:838--904, 2020.
\newblock \href {https://doi.org/10.1007/s00454-020-00243-7}
  {\path{doi:10.1007/s00454-020-00243-7}}.

\bibitem{koltun04vertical_decomposition}
Vladlen Koltun.
\newblock Almost tight upper bounds for vertical decompositions in four
  dimensions.
\newblock {\em Journal of the {ACM}}, 51(5):699--730, 2004.
\newblock \href {https://doi.org/10.1145/1017460.1017461}
  {\path{doi:10.1145/1017460.1017461}}.

\bibitem{DBLP:journals/njc/KrizancMS05}
Danny Krizanc, Pat Morin, and Michiel H.~M. Smid.
\newblock Range mode and range median queries on lists and trees.
\newblock {\em Nordic Journal of Computing}, 12:1--17, 2005.
\newblock \href {https://doi.org/10.5555/1195881.1195882}
  {\path{doi:10.5555/1195881.1195882}}.

\bibitem{law05adapt_neares_neigh_class_algor_data_stream}
Yan-Nei Law and Carlo Zaniolo.
\newblock An adaptive nearest neighbor classification algorithm for data
  streams.
\newblock In {\em Knowledge Discovery in Databases: PKDD 2005}, pages 108--120.
  Springer Berlin Heidelberg, 2005.

\bibitem{lee80two_dimen_voron_diagr_lp_metric}
D.~T. Lee.
\newblock Two-dimensional voronoi diagrams in the lp-metric.
\newblock {\em Journal of the ACM}, 27(4):604–618, oct 1980.
\newblock \href {https://doi.org/10.1145/322217.322219}
  {\path{doi:10.1145/322217.322219}}.

\bibitem{lee82voronoi}
D.~T. Lee.
\newblock On k-nearest neighbor voronoi diagrams in the plane.
\newblock {\em IEEE Transactions on Computing}, 31:478--487, 1982.
\newblock \href {https://doi.org/10.1109/TC.1982.1676031}
  {\path{doi:10.1109/TC.1982.1676031}}.

\bibitem{liu20nearl_optim_planar_neares_neigh}
Chih{-}Hung Liu.
\newblock Nearly optimal planar \emph{k} nearest neighbors queries under
  general distance functions.
\newblock In {\em Proceedings of the Thirty-First Annual {ACM-SIAM} Symposium
  on Discrete Algorithms}, pages 2842--2859. {SIAM}, 2020.

\bibitem{liu15neares_neigh_voron_diagr_revis}
Chih{-}Hung Liu, Evanthia Papadopoulou, and Der{-}Tsai Lee.
\newblock The k-nearest-neighbor voronoi diagram revisited.
\newblock {\em Algorithmica}, 71(2):429--449, 2015.
\newblock \href {https://doi.org/10.1007/s00453-013-9809-9}
  {\path{doi:10.1007/s00453-013-9809-9}}.

\bibitem{matousek93range}
Ji{\v{r}}{\'\i} Matou{\v{s}}ek.
\newblock Range searching with efficient hierarchical cuttings.
\newblock {\em Discrete \& Computational Geometry}, 10(2):157--182, 1993.
\newblock \href {https://doi.org/10.1145/142675.142732}
  {\path{doi:10.1145/142675.142732}}.

\bibitem{megiddo1979parametric}
Nimrod Megiddo.
\newblock Combinatorial optimization with rational objective functions.
\newblock {\em Mathematics of Operations Research}, 4(4):414--424, 1979.
\newblock \href {https://doi.org/10.1287/moor.4.4.414}
  {\path{doi:10.1287/moor.4.4.414}}.

\bibitem{megiddo1983parametric}
Nimrod Megiddo.
\newblock Applying parallel computation algorithms in the design of serial
  algorithms.
\newblock {\em Journal of the ACM}, 30(4):852--865, 1983.
\newblock \href {https://doi.org/10.1145/2157.322410}
  {\path{doi:10.1145/2157.322410}}.

\bibitem{DBLP:journals/comgeo/MountNSW00}
David~M. Mount, Nathan~S. Netanyahu, Ruth Silverman, and Angela~Y. Wu.
\newblock Chromatic nearest neighbor searching: {A} query sensitive approach.
\newblock {\em Computational Geometry}, 17:97--119, 2000.
\newblock \href {https://doi.org/10.1016/S0925-7721(00)00021-3}
  {\path{doi:10.1016/S0925-7721(00)00021-3}}.

\bibitem{DBLP:journals/dcg/PachS89}
J{\'{a}}nos Pach and Micha Sharir.
\newblock The upper envelope of piecewise linear functions and the boundary of
  a region enclosed by convex plates: Combinatorial analysis.
\newblock {\em Discrete \& Computational Geometry}, 4:291--309, 1989.
\newblock \href {https://doi.org/10.1007/BF02187732}
  {\path{doi:10.1007/BF02187732}}.

\bibitem{sarnak86planar}
Neil Sarnak and Robert~E Tarjan.
\newblock Planar point location using persistent search trees.
\newblock {\em Communications of the ACM}, 29(7):669--679, 1986.
\newblock \href {https://doi.org/10.1145/6138.6151}
  {\path{doi:10.1145/6138.6151}}.

\bibitem{snoeyink04pointlocation}
Jack Snoeyink.
\newblock Point location.
\newblock In {\em Handbook of Discrete and Computational Geometry, Second
  Edition}, pages 767--785. 2004.
\newblock \href {https://doi.org/10.1201/9781420035315.pt4}
  {\path{doi:10.1201/9781420035315.pt4}}.

\bibitem{tarjan83data}
Robert~Endre Tarjan.
\newblock {\em Data structures and network algorithms}.
\newblock SIAM, 1983.
\newblock \href {https://doi.org/10.1137/1.9781611970265}
  {\path{doi:10.1137/1.9781611970265}}.

\bibitem{willard85orthogonal}
Dan~E. Willard.
\newblock New data structures for orthogonal range queries.
\newblock {\em {SIAM} Journal on Computing}, 14(1):232--253, 1985.
\newblock \href {https://doi.org/10.1137/0214019} {\path{doi:10.1137/0214019}}.

\bibitem{yap90perturbation}
Chee{-}Keng Yap.
\newblock Geometric consistency theorem for a symbolic perturbation scheme.
\newblock {\em Journal of Computer and System Sciences}, 40(1):2--18, 1990.
\newblock \href {https://doi.org/10.1016/0022-0000(90)90016-E}
  {\path{doi:10.1016/0022-0000(90)90016-E}}.

\bibitem{DBLP:conf/icalp/Yu15}
Huacheng Yu.
\newblock An improved combinatorial algorithm for boolean matrix
  multiplication.
\newblock In {\em Automata, Languages, and Programming}, pages 1094--1105,
  2015.
\newblock \href {https://doi.org/10.1007/978-3-662-47672-7_89}
  {\path{doi:10.1007/978-3-662-47672-7_89}}.

\end{thebibliography}

\end{document}